\documentclass{article}
\usepackage{amsmath,amsfonts,amsthm}
 
\usepackage{thm-restate}
\usepackage{amssymb}
 \usepackage{fullpage}
\usepackage{framed}
\usepackage{bold-extra}
\usepackage{graphicx}
\usepackage{xspace}
\usepackage[auth-sc]{authblk}
\usepackage{authblk}
  \usepackage{scrextend}
\usepackage{hyperref}
\usepackage{cite}
\usepackage{enumerate}
 
\newcommand{\TT}{\mathcal{T}}
\newcommand{\GG}{\mathcal{G}}
\newcommand{\EE}{\mathcal{E}}
\newcommand{\MM}{\mathcal{M}}
\newcommand{\PP}{\mathcal{P}}
\newcommand{\II}{\mathcal{I}}

 \newcommand{\NP}{\ensuremath{\mathsf{NP}}\xspace}
\newcommand{\FPT}{\ensuremath{\mathsf{FPT}}\xspace}
\newcommand{\TempMatch}{\textsc{Temporal Matching}\xspace}
\newcommand{\TM}{\textsc{TM}\xspace}

\newcommand{\EditTempClique}{\textsc{Editing to Temporal Cliques}\xspace}
\newcommand{\ETC}{\textsc{ETC}\xspace}
\newcommand{\CompTempClique}{\textsc{Completion to Temporal Cliques}\xspace}
\newcommand{\CTC}{\textsc{CTC}\xspace}

\usepackage{tikz}
\usetikzlibrary{arrows,shapes}
\usetikzlibrary {decorations.pathmorphing, decorations.pathreplacing, decorations.shapes,
} 

\usepackage{tcolorbox}

 \newtheorem{theorem}{Theorem}[section]
 
 \newtheorem{lemma}[theorem]{Lemma}

 \newtheorem{clm}[theorem]{Claim}
 \newtheorem{remark}[theorem]{Remark}

\newenvironment{poc}[1][\unskip]{\begin{proof}[Proof of Claim~#1]}{\end{proof}}

\newcommand{\probname}{\textsc{Editing to Temporal Cliques\xspace}}

\begin{document}

\title{\textbf{A New Temporal Interpretation of Cluster Editing}\footnote{A preliminary version of this paper was presented at IWOCA 2022 \cite{BocciCMS22}.}}

 \author[1]{Cristiano Bocci}
 \author[1]{Chiara Capresi}
 \author[2]{Kitty Meeks}
 
 \author[2,3]{John Sylvester}

 \affil[1]{Dipartimento di Ingegneria dell'Informazione e Scienze Matematiche, Universit\`a degli Studi di Siena\\ 
 	
 	 \texttt{cristiano.bocci@unisi.it}, \texttt{capresi3@student.unisi.it}\medskip
}

  \affil[2]{School of Computing Science, University of Glasgow\\
  	
  	 \texttt{kitty.meeks@glasgow.ac.uk}}
     \affil[3]{Department of Computer Science, University of Liverpool\\ \texttt{john.sylvester@liverpool.ac.uk}}
\date{\vspace{-1em}}

\maketitle

\begin{abstract}
The \NP-complete graph problem \textsc{Cluster Editing} seeks to transform a static graph into a disjoint union of cliques by making the fewest possible edits to the edges.  We introduce a natural interpretation of this problem in temporal graphs, whose edge sets change over time. This problem is \NP-complete even when restricted to temporal graphs whose underlying graph is a path, but we obtain two polynomial-time algorithms for restricted cases. In the static setting, it is well-known that a graph is a disjoint union of cliques if and only if it contains no induced copy of $P_3$; we demonstrate that no general characterisation involving sets of at most four vertices can exist in the temporal setting, but obtain a complete characterisation involving forbidden configurations on at most five vertices.  This characterisation gives rise to an \FPT~algorithm parameterised simultaneously by the permitted number of modifications and the lifetime of the temporal graph.
\end{abstract}

\section{Introduction}

The \textsc{Cluster Editing} problem \cite{Ben-DorSY99} encapsulates one of the simplest and best-studied notions of graph clustering: given a graph $G$, the goal is to decide whether it is possible to transform $G$ into a disjoint union of cliques -- a \emph{cluster graph} -- by adding or deleting a total of at most $k$ edges.  While this problem is known to be \NP-complete in general \cite{CE_NP,CE_NP2,CE_NP3,CE_NP_krivanek}, it has been investigated extensively through the framework of parameterised complexity, and admits efficient parameterised algorithms with respect to several natural parameters \cite{BlasiusFGHHSWW22,cao_kernel,ChenM12,FroeseKN22,golden_ratio,above_modification,few_clusters,large_clusters,Abu-Khzam17}. The \textsc{Cluster Editing} problem was the subject of the PACE 2021 challenge \cite{Pace}. {A lot of work has been done on approximation algorithms for the problem, in this community the problem is known as correlation clustering \cite{CE_NP,AilonCN08,ChawlaMSY15,Cohen-AddadLN22}.}   

Motivated by the fact that many real-world networks of interest are subject to discrete changes over time, there has been much research in recent years into the complexity of graph problems on \emph{temporal graphs} \cite{AkridaMSR21,CasteigtsPS21,CasteigtsPS19,Erlebach0K21,MertziosMCS13}, which provide a natural model for networks exhibiting these kinds of changes in their edge-sets. One direction in particular has been to generalise static graph properties to the temporal setting \cite{AkridaCGKS19,AkridaMSZ20,HaagMNR22,MarinoS22,MertziosMZ21,ZschocheFMN20}. A first attempt to generalise \textsc{Cluster Editing} to the temporal setting was made by Chen, Molter, Sorge and Such\'y \cite{multistage}, who recently introduced the problem \textsc{Temporal Cluster Editing}: here the goal is to ensure that graph appearing at each timestep is a cluster graph, subject to restrictions on both the number of modifications that can be made at each timestep and the differences between the cluster graphs created at consecutive timesteps.  A dynamic version of the problem, \textsc{Dynamic Cluster Editing}, has also recently been studied by Luo, Molter, Nichterlein and Niedermeier \cite{dynamic}: here we are given a solution to a particular instance, together with a second instance (which will be encountered at the next timestep) and are asked to find a solution for the second instance that does not differ too much from the first.

We take a different approach, using a notion of temporal clique that already exists in the literature \cite{B-K_algorithm,maximal_cliques}.  Under this notion, a temporal clique is specified by a vertex-set and a time-interval, and we require that each pair of vertices is connected by an edge frequently, but not necessarily continuously, during the time-interval.  We say that a temporal graph is a \emph{cluster temporal graph} if it is a union of temporal cliques that are pairwise \emph{independent}: here we say that two temporal cliques are independent if either their vertex sets are disjoint, or their time intervals are sufficiently far apart (this is similar to the notion of independence used to define temporal matchings \cite{maximum_matchings}).  Equipped with these definitions, we introduce a new temporal interpretation of cluster editing, which we call \textsc{Editing to Temporal Cliques} (ETC): the goal is to add/delete a collection of at most $k$ edge appearances so that the resulting graph is a cluster temporal graph.

We prove that \ETC is \NP-hard, even when the underlying graph is a path; this reduction, however, relies on edges appearing at many distinct timesteps, and we show that, when restricted to paths, \ETC is solvable in polynomial time when the maximum number of timesteps at which any one edge appears in the graph is bounded.  It follows immediately from our hardness reduction that the variant of the problem in which we are only allowed to delete, but not add, edge appearances, is also \NP-hard in the same setting.  On the other hand, the corresponding variant in which we can only add edges, which we call \textsc{Completion to Temporal Cliques} (CTC), admits a polynomial-time algorithm on arbitrary inputs.

In the static setting, a key observation -- which gives rise to a simple \FPT~search-tree algorithm for \textsc{Cluster Editing} parameterised by the number of modifications -- is the fact that a graph is a cluster graph if and only if it contains no induced copy of the three-vertex path $P_3$ (sometimes referred to as a \emph{conflict triple} \cite{CE_survey}).  We demonstrate that cluster temporal graphs cannot be characterised by any local condition that involves only sets of at most four vertices; however, in the most significant technical contribution of this paper, we prove that a temporal graph is a cluster temporal graph if and only if every subset of at most five vertices induces a cluster temporal graph.  Using this characterisation, we obtain an \FPT~algorithm for \ETC parameterised simultaneously by the number of modifications and the lifetime (number of timesteps) of the input temporal graph.

\subsection{Related Work}

\textsc{Cluster Editing} is known to be \NP-complete \cite{CE_NP,CE_NP2,CE_NP3,CE_NP_krivanek}, even for graphs with maximum degree six and when at most four edge modifications incident to each vertex are allowed \cite{loc_bound}.  On the positive side, the problem can be solved in polynomial time if the input graph has maximum degree two \cite{golden_ratio} or is a unit interval graph \cite{unit_int}.  Further complexity results and heuristic approaches are discussed in a survey article \cite{CE_survey}.

Variations of the problem in which only deletions or additions of edges respectively are allowed have also been studied.  The version in which edges can only be added is trivially solvable in polynomial time, since an edge must be added between vertices  $u$ and $v$ if and only if $u$ and $v$ belong to the same connected component of the input graph but are not already adjacent.  The deletion version, on the other hand, remains \NP-complete even on $4$-regular graphs, but is solvable in polynomial time on graphs with maximum degree three \cite{loc_bound}.

\textsc{Cluster Editing} has received substantial attention from the parameterised complexity community, with many results focusing on the natural parameterisation by the number $k$ of permitted modifications.  Fixed-parameter tractability with respect to this parameter can easily be deduced from the fact that a graph is a cluster graph if and only if it contains no induced copy of the three-vertex path $P_3$, via a search tree argument; this approach has been refined repeatedly in non-trivial ways, culminating in an algorithm running in time $\mathcal{O}(1.62^k + m + n)$ for graphs with $n$ vertices and $m$ edges \cite{golden_ratio}.  More recent work has considered as a parameter the number of modifications permitted above the lower bound implied by the number of modification-disjoint copies of $P_3$ (copies of $P_3$ such that no two share either an edge or a non-edge) \cite{above_modification}.  Other parameters that have been considered include the number of clusters \cite{few_clusters} and a lower bound on the permitted size of each cluster \cite{large_clusters}.

\subsection{Organisation of the Paper}
We begin in Section \ref{sec:preliminaries} by introducing some notation and definitions, and formally defining the \ETC problem. In Section \ref{sec:generallems} we collect several results and fundamental lemmas which are either used in several other sections or may be of independent interest. Section \ref{sec:paths} concerns restricting the problem to temporal graphs which have a path as the underlying graph: in Section \ref{sec:pathhard} we show that \ETC is \NP-hard even in this setting, however in Section \ref{sec:pathspoly} we then show that if we further restrict temporal graphs on paths to only have a bounded number of appearances of each edge then \ETC is solvable in polynomial time. In Section \ref{sec:additiononly} we consider a variant of the \ETC problem where edges can only be added, and show that this can be solved in polynomial time on any temporal graph. Finally, in Section \ref{sec:characterisation} we present our main result which gives a characterisation of cluster temporal graphs by induced temporal subgraphs on five vertices. We prove this result in Section \ref{sec:charproof} before applying it in Section \ref{sec:charalg} to show that \ETC is in \FPT~when parameterised by the lifetime of the temporal graph and number of permitted modifications.

\section{Preliminaries}\label{sec:preliminaries}

In this section we first give some basic definitions and introduce many new notions that are key to the paper, before formally specifying the \ETC problem.

\subsection{Notation and Definitions}\label{sec:notation}

\paragraph{Elementary Definitions.} Let $\mathbb{N}$ denote the natural numbers (with zero) and $\mathbb{Z}^{+}$ denote the positive integers. 
We refer to a set of consecutive natural numbers $[i,j]=\{i,i+1,\dots ,j\}$ for some $i,j\in \mathbb{N}$ with $i\leq j$ as an \textit{interval}, and to the number~$j-i+1$ as the \textit{length} of the interval. Given an undirected (static) graph $G=(V,E)$ we denote its vertex-set by $V=V(G)$ and edge-set by $E=E(G)\subseteq \binom{V}{2}$. For two {edges} $ e=xy $ and $e'= wz $ we slightly abuse notation by taking $e\cap e'$ to mean $\{x,y\} \cap \{w,z\}$ as this is more convenient. We say that $b_n=\mathcal{O}(a_n)$ if there exist constants $C$ and $n_0$ such that $|b_n| \leq C\cdot a_n$ for $n\geq n_0$.  We work in the word RAM model of computation, so that arithmetic operations on integers represented using a number of bits logarithmic in the total input size can be carried out in time $\mathcal{O}(1)$.

\paragraph{Temporal Graphs.} A {\it temporal graph} $\mathcal{G}=(G,\TT)$ is a pair consisting of a static (undirected) underlying graph $G=(V,E)$ and a labeling function $\TT:E \to 2^{\mathbb{N}}\setminus \{\emptyset\}$. For a static edge $e\in E$, we think of $\mathcal{T}(e)$ as the set of time appearances of $e$ in $\mathcal{G}$ and let $\mathcal E(\mathcal G) := \{ (e,t) :  e \in E$ and $t \in \TT(e) \}$ denote the set of edge appearances, or \textit{time-edges}, in a temporal graph $\mathcal G$.  We consider temporal graphs $\mathcal{G}$ with \emph{finite
lifetime} given by $T(\mathcal G):=\max \{t\in \TT(e):  e\in E\}$, that is, there is a maximum label assigned by~$\TT$ to an
edge of~$G$. {Unless otherwise stated,} we assume w.l.o.g.\ that $\min \{t\in \TT(e):  e\in E\}=1$. We denote the lifetime of $\mathcal G$ by $T$ when $\mathcal{G}$ is clear from the context. The \emph{snapshot}  of $\mathcal{G}$ \emph{at time}~$t$ is the static graph on $V$ with edge set $\{e\in E:  t\in \TT (e)\}$. Given temporal graphs $\mathcal{G}_1$ and $\mathcal{G}_2$, let $\mathcal{G}_1 \triangle \mathcal{G}_2$ be the set of time-edges appearing in exactly one of $\mathcal{G}_1$  or $\mathcal{G}_2$.  For the purposes of computation, we assume that $\mathcal{G}$ is given as a list of (static) edges together with the list of times $\mathcal{T}(e)$ at which each static edge appears, so the size of $\mathcal{G}$ is $\mathcal{O}(\max\{|\mathcal{E}|,|V|\}) = \mathcal{O}(|V|^2T)$. 

\paragraph{Templates and Cliques.} For an edge $e\in E(G)$ in the underlying graph of a temporal graph $\mathcal{G}=(G,\TT)$, an interval $[a,b]$ {with $a\leq  b$}, and $\Delta_1 \in \mathbb{Z}^+$, we say that $e$ is \textit{$\Delta_1$-dense} in $[a,b]$ {if $\mathcal{T}(e) \cap [\tau,\tau + \Delta_1-1]\neq \emptyset $ for all $\tau \in [a,\max\{a,b-\Delta_1+1\}]$}. We define a {\it template} to be a pair $C=(X,[a,b])$  where $X$ is a set of vertices and $[a,b]$ is an interval. {We restrict templates to be non-empty by requiring $a\leq b$ and $|X|\geq 2$, this is a technical restriction but nothing is lost when we come to define temporal cliques and this avoids the existence of empty templates appearing in some of our definitions.}  For a set $S$ of time-edges we let $V(S)$ denote the set of all endpoint of time-edges in $S$, and the \textit{lifetime} $L(S)=[s,t]$, where $s=\min\{s: (e,s)\in S\}$ and $t=\max\{t: (e,t)\in S\}$. Furthermore, we say that $S$ {\it generates} the template $(V(S),L(S))$. A set $S\subset \mathcal{E}(\mathcal{G})$ forms a \textit{$\Delta_1$-temporal clique} if for every pair $x,y\in V(S)$ of vertices in the template $(V(S),L(S))$ generated by $S$, the edge $xy$ is $\Delta_1$-dense in $L(S)$. 

{The idea is to capture the notion of vertices being associated but without the strong restriction that a temporal clique must be a clique in every time step. For an example, one may think of a `clique' of school friends where each vertex is a child and a time-edge represents an interaction between two children on that day. If we take one term/semester as the lifetime of the clique then the children may not see each other at the weekends or the half-term break, but since these breaks from interaction are not so long (say at most a week) then if we would set $\Delta_1$ to be a week then the children would be considered as a $\Delta_1$-temporal clique.}

\paragraph{Independence and Cluster Temporal Graphs.} For $\Delta_2 \in \mathbb{Z}^+$ we say that two templates $(X,[a,b])$ and $(Y,[c,d])$ are \textit{$\Delta_2$-independent} if \[X\cap Y=\emptyset\quad \text{ or } \quad\min_{s \in [a,b], t\in [c,d]}|s-t|\geq \Delta_2.\] Let $\mathfrak{T}(\mathcal{G},\Delta_2)$ denote the class of all collections of pairwise $\Delta_2$-independent templates on $V(\GG)$ where templates in the collection have lifetimes $[a,b]$ satisfying $1\leq a\leq b\leq T(\mathcal{G})$. Two sets $S_1,S_2$ of time-edges are \textit{$\Delta_2$-independent} if the templates they generate are $\Delta_2$-independent. As a special case of this, two time-edges $(e,t)$, $(e',t')$ are \textit{$\Delta_2$-independent} 
if $e \cap e' = \emptyset$ or~$|t - t'| \geq \Delta_2$. {This definition of $\Delta_2$-independence was introduced (for time-edges) in \cite{maximum_matchings} and, in our opinion, is the most natural way to relax total independence of time-edges (not sharing any endpoints).} 

A temporal graph $\mathcal{G}$ \textit{$\Delta_1$-realises a collection} $\{(X_i,[a_i,b_i])\}_{i\in[k]}\in \mathfrak{T}(\mathcal{G},\Delta_2)$ of pairwise $\Delta_2$-independent templates if 
\begin{itemize}
\item for each $(xy, t)\in \mathcal{E}$ there exists $i\in [k]$ such that $x,y\in X_i$ and $t\in[a_i,b_i]$,
    \item for each $i\in [k]$ and $x,y\in X_i$, the edge $xy$ is $\Delta_1$-dense in $[a_i,b_i]$. 
\end{itemize}
If there exists some $\mathcal{C}\in \mathfrak{T}(\mathcal{G},\Delta_2)$ such that $\mathcal{G}$ $\Delta_1$-realises $\mathcal{C}$ then we call $\mathcal{G}$ a $(\Delta_1,\Delta_2)$-\textit{cluster temporal graph}. \begin{remark}\label{rem:deltas}Throughout we assume that $\Delta_2> \Delta_1$, since if $\Delta_2\leq \Delta_1$ then one $\Delta_1$-temporal clique can realise many different sets of $\Delta_2$-independent templates. For example, if $\Delta_2=\Delta_1$ then the two time-edges $(e,t)$ and $(e,t+\Delta_1)$ are $\Delta_2$-independent but also $e$ is $\Delta_1$-dense in the interval $[t,t+\Delta_1]$. \end{remark}

\paragraph{Induced, Indivisible, and Saturated Sets.}  Let $S$ be a set of time-edges and $A$ be a set of vertices, then we let $S[A]=\{(xy,t) \in S : x,y\in A \}$ be the set of all the time-edges in $S$ \textit{induced} by $A$. Similarly, given a temporal graph $\mathcal{G}$ and $A \subset V$, we let $\mathcal{G}[A]$ be the temporal graph with vertex set $A$ and temporal edges $\mathcal{E}[A]$. For an interval $[a,b]$ we let $\mathcal{G}|_{[a,b]}$ be the temporal graph on $V$ with the set of time-edges $\{(e,t)\in \mathcal{E}(\mathcal{G}) : t \in [a,b]\}$. We will say that a set $S$ of time-edges is {\it $\Delta_2$-indivisible} if there does not exist a pairwise $\Delta_2$-independent collection $\{S_1,\dots, S_k\}$ of time-edge sets satisfying $\cup_{i\in[k]}S_i =S$. A $\Delta_2$-indivisible set $S$ is said to be $\Delta_2$-\textit{saturated in $\mathcal{G}$} if it is {inclusion-maximal, that is,} after including any other time-edge of $\mathcal{E}(\mathcal G)$ it would cease to be $\Delta_2$-indivisible.

\subsection{Problem Specification}\label{sec:probspec}

\paragraph{Editing to Temporal Cliques.} We can now introduce the \ETC problem which, given as input a temporal graph $\mathcal G$ on a vertex set $V$ and natural numbers $k, \Delta_1, \Delta_2 \in \mathbb{Z}^+$, asks whether it is possible to transform $\mathcal G$ into a $(\Delta_1,\Delta_2)$-cluster temporal graph by applying at most $k$ modifications (addition or deletion) of time-edges {within the lifetime of $\mathcal{G}$}. We now give a formal definition of this problem.

\begin{framed}
	\noindent
 \EditTempClique (\ETC):\\ 
\textit{Input:} A temporal graph $\mathcal{G}=(G,\mathcal{T})$ and positive integers $k,\Delta_1,\Delta_2 \in \mathbb{Z}^+$.\\
\textit{Question:} { Does there exist a set $\Pi \subseteq \binom{V}{2}\times [T(\mathcal{G})]$ with $ |\Pi|\leq k$ such that the temporal graph on $V$ with time-edges $\mathcal{E}(\mathcal{G}) \triangle \Pi$ is a $(\Delta_1,\Delta_2)$-cluster temporal graph?} 
\end{framed} 

We begin with a simple observation on the hardness of \ETC which shows {that} we can only hope to gain tractability in settings where the static version is tractable. However, we shall see in Section \ref{sec:pathhard} that \ETC is hard on temporal graphs with the path as the underlying graph, and thus the converse does not hold.

\begin{restatable}{proposition}{prop:NPhard}\label{prop:NPhard}
Let $\mathcal{C}$ be a class of graphs on which \textsc{Cluster Editing} is \NP-complete.  Then \ETC is \NP-complete on the class of temporal graphs $\{(G,\mathcal{T}): G \in \mathcal{C}\}$.
\end{restatable}
 
\begin{proof}[Proof of Proposition \ref{prop:NPhard}]
We obtain a reduction from \textsc{Cluster Editing} to \ETC as follows.  Given an instance $(G,k)$ of \textsc{Cluster Editing}, where $G \in \mathcal{C}$, we construct an instance $(\mathcal{G},1,1,k)$ of \ETC by setting $\mathcal{G} = (G,\mathcal{T})$ where $\mathcal{T}(e) = \{1\}$ for all $e \in E(G)$.  It is straightforward to verify that a set $\Pi$ of time-edge modifications makes $\mathcal{G}$ a $(1,1)$-cluster temporal graph if and only if the set of edge modifications $\{e: (e,1) \in \Pi\}$ modifies $G$ into a cluster graph.
  \end{proof}

\section{General Observations about \ETC}\label{sec:generallems}
In this section we collect many fundamental results on the structure of temporal graphs and the cluster editing problem which we will need later in the paper. {Many of these results concern $\Delta_2$-saturated sets, the results in this section (and Section \ref{sec:additiononly}) indicate that $\Delta_2$-saturated take something of the role of connected components in the static setting.} {We begin with three elementary lemmas.}  

\begin{lemma}\label{lem:IndivisibleIntersection}
If two $\Delta_2$-indivisible sets $S_1$ and $S_2$ of time-edges satisfy $S_1\cap S_2\neq \emptyset$, then $S_1 \cup S_2$ is $\Delta_2$-indivisible.
\end{lemma}

\begin{proof}
Let us assume by contradiction that $S_1 \cup S_2$ is not $\Delta_2$-indivisible. Thus, there exist a partition of $S_1 \cup S_2$ into at least two {non-empty} $\Delta_2$-independent subsets. Now, because $S_1$ and $S_2$ are respectively $\Delta_2$-indivisible, then each of them must belong entirely to a same subset of the partition. However, since $S_1\cap S_2\neq \emptyset$, then $S_1$ and $S_2$ must belong both to the same set of the partition. Thus, this contradicts our assumption that the partition contains more than one set.
  \end{proof}

\begin{lemma} \label{sat_set}
Let $\mathcal G$ be a temporal graph, $S \subseteq \mathcal E(\mathcal G)$ be a $\Delta_2$-saturated set of time-edges, and $\mathcal K$ a $\Delta_1$-temporal clique such that $\mathcal K \subseteq \mathcal E(\mathcal G)$ and $\mathcal K \cap S \not = \emptyset$. Then, $\mathcal K \subseteq S$.
\end{lemma}

\begin{proof} Because $\mathcal K$ is a $\Delta_1$-temporal clique it forms a $\Delta_2$-indivisible set. Then as $S$ is $\Delta_2$-indivisible and has non-empty intersection with $\mathcal{K}$ we see that $\mathcal{K}\cup S$ is $\Delta_2$-indivisible by Lemma \ref{lem:IndivisibleIntersection}. If $\mathcal{K}\not\subseteq S$ then this contradicts the assumption that $S$ is an {inclusion-maximal} $\Delta_2$-indivisible set {of time-edges}.
  \end{proof}

\begin{lemma}\label{lem:indivisibility}
Let $\mathcal{G}$ be any $(\Delta_1,\Delta_2)$-cluster temporal graph.  Then, any $\Delta_2$-indivisible set $S\subseteq \mathcal{E}(\mathcal{G})$ must be contained within a single $\Delta_1$-temporal clique.
\end{lemma}
\begin{proof}Let $C_1, \dots, C_q$ be the $\Delta_1$-temporal cliques in the $(\Delta_1,\Delta_2)$-cluster temporal graph $\mathcal{G}$ in {an arbitrary} decomposition of $\mathcal{G}$ into $\Delta_1$-temporal cliques such that $C_i \cap S \neq \emptyset$, { for all $i\in[q]$}. Now, for each $i\in[q]$ we let $S_i'= S\cap C_i $. Observe that since $S_i'\subseteq C_i$ then the sets $\{S_i'\}_{i\in[q]}$ are pairwise $\Delta_2$-independent and $\cup_{i\in [q]} S_i' =S$, {if $q>1$ then }this contradicts the indivisibility of $S$.
  \end{proof}

Our {next} result shows that there is a way to uniquely partition any temporal graph.

\begin{lemma} \label{saturated}
{Let $\mathcal{G}$ be any temporal graph and $\Delta_2\in \mathbb{Z}^+$. Then, $\mathcal{G}$ }has a unique decomposition of its time-edges into $\Delta_2$-saturated subsets.
\end{lemma}

\begin{proof}
Let us consider any temporal graph $\mathcal G$. We will first argue that it is always possible to find a partition of $\mathcal E(\mathcal G)$ into sets $\Delta_2$-saturated in $\mathcal G$.  Note that, given any time-edge, we can greedily build a $\Delta_2$-saturated set that contains it. We now claim that any time-edge $(e,t) \in \mathcal E(\mathcal G)$ can belong to exactly one $\Delta_2$-saturated subset of $\mathcal E(\mathcal G)$. To see this, suppose that $P$ and $P'$ are two distinct $\Delta_2$-saturated sets with $(e,t) \in P \cap P'$.  By Lemma \ref{lem:IndivisibleIntersection}, we then have $P \cup P'$ is also $\Delta_2$-indivisible, contradicting {the assumption that $P$ and $P'$ are $\Delta_2$-saturated}. Thus, if $\{P_1,\dots,P_h\}$ is any collection of $\Delta_2$-saturated subsets such that $\bigcup_{i=1}^h P_i = \mathcal{E}(\mathcal{G})$, we know that for $P_i \neq P_j$ we have $P_i \cap P_j = \emptyset$ and hence $\mathcal{P} = \{P_1,\dots,P_h\}$ is a partition of the time-edges into $\Delta_2$-saturated sets.

We now prove uniqueness of this decomposition.  In fact, if there was a different partition $\mathcal{P}'$ of the time-edges of $\mathcal{G}$ into $\Delta_2$-saturated sets, it would imply that there exists at least one time-edge $(e,t) \in \mathcal E(\mathcal G)$ which belongs to different $\Delta_2$-saturated sets in $\mathcal{P}$ and $\mathcal{P'}$.  But in this case, $(e,t)$ belongs to two different $\Delta_2$-saturated subsets of $\mathcal E(\mathcal G)$ which we already proved to be impossible.  
  \end{proof}

{We now show that} this partition can be found in polynomial time. 
\begin{lemma}\label{lem:find-saturated} 
Let $\mathcal{G} = (G = (V,E), \mathcal{T})$ be a temporal graph, and let $\mathcal{E} = \{(e,t): e \in E, t \in \mathcal{T}(e)\}$ be the set of time-edges of $\mathcal{G}$. Then, there is an algorithm which finds {and outputs} the unique partition of $\mathcal{E}$ into $\Delta_2$-saturated subsets in time $\mathcal{O}(|\mathcal{E}|^3|V|)$.
\end{lemma}
\begin{proof}
Our algorithm proceeds as follows.  We maintain a list of disjoint subsets of time-edges, and make multiple passes through the list until it contains precisely the $\Delta_2$-saturated subsets of $\mathcal{E}$.  Suppose that $\mathcal{E} = \{e_1,\dots,e_r\}$; we initialise $\mathcal{L}_0$ as the list {$(E_i)_{i=1}^r$, where $E_i=\{e_i\}$ for each $i\in [r]$.} For convenience we will also store the template generated by each element of the list.  In a single pass through the list $\mathcal{L}_i$, we remove the first element $E_1$ of $\mathcal{L}_i$ and consider each remaining element of the list in turn: as soon as we find a set $E_j$ such that the templates $(X_1,[a_1,b_1])$ and $(X_j,[a_j,b_j])$ generated by $E_1$ and $E_j$ respectively are not $\Delta_2$-independent, we remove $E_j$ from $\mathcal{L}_i$, add the set $E_1 \cup E_j$ to $\mathcal{L}_{i+1}$ (noting that we also compute and store the template generated by this set of time-edges), and move on to $E_2$; if we find no such set $E_j$, we add $E_1$ to $\mathcal{L}_{i+1}$.  Now proceed in the same way with the first element of the modified list $\mathcal{L}_i$; we continue until $\mathcal{L}_i$ is empty.  Assuming we combined at least one pair of sets when moving from $\mathcal{L}_i$ to $\mathcal{L}_{i+1}$, we now repeat the entire procedure with $\mathcal{L}_{i+1}$.  When we reach an index $\ell$ such that we complete a pass through $\mathcal{L}_{\ell}$ without carrying out any merging operations, we halt and claim that $\mathcal{L}_{\ell}$ {is a partition of $\mathcal{E}$ into $\Delta_2$-saturated sets}.

To see that this claim holds, we first argue that the following invariant holds throughout the execution of the algorithm: for each $i$, every set of time-edges in $\mathcal{L}_i$ is $\Delta_2$-indivisible.  

This invariant clearly holds for $\mathcal{L}_0$, since each element of $\mathcal{L}$ contains only a single time-edge.  Suppose now that the invariant holds for $\mathcal{L}_i$; we will argue that it must also hold for $\mathcal{L}_{i+1}$.  Let $E_r$ and $E_s$ be two sets that are merged during the pass through $\mathcal{L}_i$.  It suffices to demonstrate that $E_r \cup E_s$ is $\Delta_2$-indivisible.  We already know, from the fact that the invariant holds for $\mathcal{L}_i$, that there is no way to partition $E_1$ (respectively $E_2$) into two or more pairwise $\Delta_2$-independent sets; thus, the only possible way to partition $E_1 \cup E_2$ into pairwise $\Delta_2$-independent sets is with the partition $(E_1,E_2)$.  However, by construction, since the algorithm merged $E_1$ and $E_2$, we know that $E_1$ and $E_2$ cannot be $\Delta_2$-independent.  We conclude that the invariant also holds for $\mathcal{L}_{i+1}$.

By the termination condition for the algorithm, we know that $\mathcal{L}_{\ell}$ gives a partition of $\mathcal{E}$ into pairwise $\Delta_2$-independent sets.  Moreover, the invariant tells us that each of these sets is $\Delta_2$-indivisible.  Maximality and hence saturation of the sets follows immediately from the fact that the edges are partitioned into pairwise $\Delta_2$-independent sets, so we conclude that $\mathcal{L}_{\ell}$ does indeed give a partition of $\mathcal{E}$ into pairwise $\Delta_2$-independent $\Delta_2$-saturated sets.  Finally, by Lemma \ref{saturated}, we know that this partition is unique.

It remains only to bound the running time of the algorithm.  Note that a single pass (in which we create $\mathcal{L}_{i+1}$ from $\mathcal{L}_i$ requires us to consider $\mathcal{O}(|\mathcal{L}_i|^2)$ pairs of sets; since we store the templates generated by each set, each such comparison can be carried out in time $\mathcal{O}(V)$ as this will be dominated by the time required to verify disjointness or otherwise of the vertex sets.  (Note that we can compute the template for each new merged set in constant time.)  Since $|\mathcal{L}_0| = |\mathcal{E}|$ and the length of the list decreases by at least one with each pass, we see that the overall running time is bounded by $\mathcal{O}(|\mathcal{E}|^3|V|)$, as required. 
\end{proof}

Since any temporal graph has a unique decomposition into $\Delta_2$-saturated sets by Lemma \ref{saturated} and any pair of $\Delta_2$-saturated sets is $\Delta_2$-independent by definition,  we obtain the following result.  

\begin{lemma} \label{sat_cliques}
A temporal graph $\mathcal G$ is a $(\Delta_1,\Delta_2)$-cluster temporal graph if and only if every $\Delta_2$-saturated set of time-edges forms a $\Delta_1$-temporal clique.
\end{lemma}

\begin{proof}

[$\implies$] Let us assume that $\mathcal G$ is a $(\Delta_1,\Delta_2)$-cluster temporal graph and let $\mathcal P_{\mathcal G}$ be the decomposition of its time-edges into $\Delta_2$-saturated subsets. Because any $\Delta_2$-saturated set is also $\Delta_2$-indivisible by definition, then by Lemma \ref{lem:indivisibility} it follows that for any of these $\Delta_2$-saturated sets $S$ there exists a $\Delta_1$-temporal clique $\mathcal K$ in $\mathcal G$ which contains $S$ in its time-edge set. This means that $\mathcal K \cap S \not =\emptyset$, so by Lemma \ref{sat_set}, {$\mathcal K \subseteq S$}. So, since  $S$ is $\Delta_2$-saturated, $S$ is in fact the $\Delta_1$-temporal clique $\mathcal{K}$.

[$\impliedby$] By Lemma \ref{saturated} there exists a unique decomposition $\mathcal P_{\mathcal G}$ of $\mathcal G$ into $\Delta_2$-saturated sets. {By hypothesis each of these $\Delta_2$-saturated sets forms a $\Delta_1$-temporal clique and so it remains to show that they are pairwise $\Delta_2$-independent.} To see this, suppose, for a contradiction, that two $\Delta_2$-saturated sets $X$ and $Y$ are not $\Delta_2$-independent.  By indivisibility of $X$ and $Y$, the only possible way in which $X \cup Y$ could be partitioned into pairwise $\Delta_2$-independent subsets is with the partition $(X,Y)$, but by assumption $X$ and $Y$ are not $\Delta_2$-independent, so we can conclude that $X \cup Y$ is in fact $\Delta_2$-indivisible.  This contradicts maximality of $X$ and $Y$, giving the contradiction to the assumption that these are $\Delta_2$-saturated sets.  Hence, if all of these sets are $\Delta_1$-temporal cliques, trivially follows that $\mathcal G$ is a $(\Delta_1,\Delta_2)$-cluster temporal graph.
  \end{proof}

We can now prove what may be considered as the main standalone result in this section, this shows that $(\Delta_1,\Delta_2)$-cluster temporal graphs can be recognised in polynomial time.

\begin{lemma}\label{prop:verify-poly}
Let $\mathcal{G} = {(G,\mathcal{T})}$ be a temporal graph, and  $\mathcal{E} = \{(e,t): e \in E(G), \; t \in \mathcal{T}(e)\}$ be the set of time-edges of $\mathcal{G}$.  Then, we can determine in time $\mathcal{O}(|\mathcal{E}|^3|V|)$ whether $\mathcal{G}$ is a $(\Delta_1,\Delta_2)$-cluster temporal graph.
\end{lemma}
\begin{proof}
We know by Lemma \ref{sat_cliques} that $\mathcal{G}$ is a $(\Delta_1,\Delta_2)$-cluster temporal graph if and only if every \mbox{$\Delta_2$-saturated} set of edges forms a $\Delta_1$-temporal clique.  By Lemma \ref{lem:find-saturated} we can find the unique partition of $\mathcal{E}$ into \mbox{$\Delta_2$-saturated} subsets in time $\mathcal{O}(|\mathcal{E}|^3|V|)$; our algorithm to do this can easily be adapted so that it also outputs the template generated by each of the edge-sets.  To determine whether each edge-set forms a \mbox{$\Delta_1$-temporal} clique, it suffices to consider each template $(X,[a,b])$ in turn, and for each pair $x,y \in X$ to verify whether the edge $xy$ is $\Delta_1$-dense in the interval $[a,b]$.  It is clear that we can do this for all templates in time $\mathcal{O}(|\mathcal{E}|)$.
  \end{proof}

The following lemma gives the `easy' direction of the Characterisation proved in Section \ref{sec:charproof}. 

\begin{lemma} \label{easychar}
Let $\mathcal G$ be a $(\Delta_1,\Delta_2)$-cluster temporal graph and $S \subseteq V(\mathcal G)$. Then, $\mathcal G[S]$ is also a $(\Delta_1,\Delta_2)$-cluster temporal graph.
\end{lemma}

\begin{proof} Let $\mathcal{K}$ be an arbitrary $\Delta_1$-temporal clique of $\mathcal{G}$. We shall show that $\mathcal{K}[S]$ is a $\Delta_1$-temporal clique. The result will then follow since $\mathcal{G}$ consists of a union of $\Delta_2$-independent $\Delta_1$-temporal cliques and two such $\Delta_1$-temporal cliques will still be $\Delta_2$-independent after restricting to $S$ since $\mathcal{K}[S]\subseteq \mathcal{K}$. 

Recall that $\mathcal{K}[S] = \{(xy, t) \in \mathcal{K}: x,y\in S \} $ and let $C=(X,[a,b])$ be the (tight) template realised by the set of time-edges $\mathcal{K}[S]$. It remains to show that for all $x,y\in X$, the edge $xy$ is $\Delta_1$-dense in $[a,b]$.  Assume otherwise, then we can find some $x,y \in X=V(\mathcal{K})\cap S$ and an interval $[c,d]\subseteq [a,b]$ of length at least $\Delta_1$ with no appearance of $xy$. However, since $x,y \in V(\mathcal{K})$ and $[a,b]\subseteq L(\mathcal{K})$ this contradicts the fact that $\mathcal{K}$ is a $\Delta_1$-temporal clique.  
  \end{proof}

Recall $\mathcal{G}_1 \triangle \mathcal{G}_2$ denotes the set of time-edges appearing in exactly one of the temporal graphs $\mathcal{G}_1$  or $\mathcal{G}_2$.

\begin{lemma} \label{induced_connectivity}
Let $\mathcal G=(G,\mathcal T)$ be a temporal graph, and $\mathcal C  \in \mathfrak{T}(\mathcal{G},\Delta_2)$ minimise \[\min_{\mathcal{G}_{\mathcal C} \text{ $\Delta_1$-realises } \mathcal C} |\mathcal{G} \triangle \mathcal{G}_{\mathcal C}|.\] Then, for any template $C=(X,[a,b]) \in \mathcal C$, the static underlying graph of $\mathcal{G}[X]|_{[a,b]}$ is connected.
\end{lemma}

\begin{proof}
	Let us assume that for some $C=(X,[a,b]) \in \mathcal C$, the graph $G[X]|_{[a,b]}$ is disconnected. This implies that there exists a set $Y \subset V(G)$ such that $Y \neq \emptyset$, $Y \subset X$ and for any $x\in X\backslash Y $ and $y\in Y$, the static edge $xy$ does not belong to the underlying graph $G[X]|_{[a,b]}$. Because $X$ is the vertex set of the  template $C=(X,[a,b])$, which by definition is induced by a $\Delta_1$-temporal clique in $\mathcal G_{\mathcal C}$, the set \[\Pi=\big\{(xy,t): x\in X\backslash Y, \; y\in Y, \; t \in [a,b] \big\}\subseteq \mathcal E(\mathcal G) \triangle \mathcal E(\mathcal G_{\mathcal C}),\] satisfies $|\Pi| \geq \max \{1, \lfloor \frac{b-a}{\Delta_1} \rfloor\} \geq 1$. 
	
	We now construct two new templates $C^{ Y}=(Y, [a,b])$ and $C^{X \backslash Y}=(X \setminus Y, [a,b])$ and define the collection $\mathcal C'  = (\mathcal C \backslash C)\cup C^{ Y} \cup C^{ X\backslash Y} $. Observe that \[\mathcal E \left(\mathcal G_{\mathcal C}[X]|_{[a,b]}\right)=\mathcal E\left(\mathcal G_{\mathcal C'}\left[Y\right]|_{[a,b]} \right) \cup \mathcal E\left(\mathcal G_{\mathcal C'}\left[  X\backslash Y \right]|_{[a,b]}\right) \cup \Pi ,\] where the union is disjoint and all other edges in $\mathcal{G}_{\mathcal C'} $ are also contained in $\mathcal{G}_{\mathcal C}$. Thus, $|\mathcal E(\mathcal G) \triangle \mathcal E(\mathcal G_{\mathcal C'}) |< |\mathcal E(\mathcal G) \triangle \mathcal E(\mathcal G_{\mathcal C})|$, contradicting the minimality of $ \mathcal C$.
  \end{proof}

Finally, we show that there always exists a solution to the \EditTempClique problem with lifetime no greater than that of the input temporal graph.

\begin{lemma}\label{prop:within-lifetime}
For any temporal graph $\mathcal{G}$, there exists a $(\Delta_1,\Delta_2)$-cluster temporal graph $\mathcal{G}'$, minimising the edit distance $|\mathcal{G} \triangle \mathcal{G}'|$ between $\mathcal{G}$ and $\mathcal{G}'$, such that the lifetime of $\mathcal{G}'$ is a subset of the lifetime of $\mathcal{G}$.
\end{lemma}
\begin{proof}
Let $\widetilde{\mathcal{G}}$ be any $(\Delta_1,\Delta_2)$-cluster temporal graph such that $|\widetilde{\mathcal{G}} \triangle \mathcal{G}|$ is minimum.  {Let $I=[s,t]$ be the smallest interval containing all edges of $\mathcal{G}$}.  It suffices to demonstrate that we can transform $\widetilde{\mathcal{G}}$ into a $(\Delta_1,\Delta_2)$-cluster temporal graph $\mathcal{G}'$ whose lifetime is contained in $I$, without increasing the edit distance.  Specifically, if $\mathcal{C} = \{C_1,\dots,C_m\}$ is the set of clique templates generated by $\widetilde{\mathcal{G}}$, we will demonstrate that, for any template $C_i = (X_i,[a_i,b_i])$, we can modify $\widetilde{\mathcal{G}}$ so that its restriction to $X_i$ and $[a_i,b_i]$ is a clique generating the template $(X_i,[a_i,b_i] \cap I)$, without increasing the edit distance.  Notice that modifying each clique in this way cannot violate the $\Delta_2$-independence of the cliques, since we only decrease the lifetime of each clique.

Suppose now that $(X_i,[a_i,b_i]) \in \mathcal{C}$, and $[a_i,b_i] \not \subseteq I$.  Set $\mathcal{D}_1$ to be the set of time-edges in $\widetilde{\mathcal{G}}[X_i]|_{[a_i,s-1]}$, and $\mathcal{D}_2$ to be the set of time-edges in $\widetilde{\mathcal{G}}[X_i]|_{[t+1,b_i]}$; by assumption at least one of $\mathcal{D}_1$ and $\mathcal{D}_2$ is non-empty.  Note that no time-edges in $\mathcal{D}_1 \cup \mathcal{D}_2$ appear in $\mathcal{G}$.  Now set $\mathcal{G}'$ to be the temporal graph obtained from $\widetilde{\mathcal{G}}$ by deleting all time-edges in $\mathcal{D}_1 \cup \mathcal{D}_2$ and adding the set of time-edges 
\[\{(e,s): \exists r \text{ with } (e,r) \in \mathcal{D}_1\} \cup \{(e,t): \exists r \text{ with } (e,r) \in \mathcal{D}_2\}.\]  
Note that we have $|\mathcal{G}' \triangle \mathcal{G}| \le |\widetilde{\mathcal{G}} \triangle \mathcal{G}|$, so the edit distance does not increase, and moreover  $\mathcal{G}'[X_i]|_{[a_i,b_i]}$ does not contain any edge-appearances at times not in $I$.  It remains only to demonstrate that the edges of $\mathcal{G}'[X_i]|_{[a_i,b_i] \cap [T]}$ form a $\Delta_1$-temporal clique, that is, that for every pair of vertices $x,y \in X_i$ the edge $xy$ is $\Delta_1$-dense in $[a_i,b_i] \cap I$.  Since we did not remove any edge-appearances in this interval, the only way this could happen is if $xy$ did not appear at all in this interval in $\widetilde{\mathcal{G}}$; in this case, we must have $\min\{b_i,t\} - \max\{s,a_i\} < \Delta_1 - 1$, and so a single appearance of the edge $xy$ in this interval -- which is guaranteed by construction -- is enough to give the required $\Delta_1$-density.
  \end{proof}

\section{ETC on Paths}\label{sec:paths}

Throughout $P_n$ will denote the path on $V(P_n)=\{v_1,\dots , v_n\}$ with $E(P_n)= \{v_{i}v_{i+1} : 1 \leq i < n \}$. {Define \[\mathfrak{F}_n=\left\{\mathcal{P}_n : \mathcal{P}_n = (P_n,\mathcal{T}) , \,\TT:E \to 2^{\mathbb{N}}\setminus \{\emptyset\},\, \min \{t\in \TT(e):  e\in E\}=1 \right\}\] to be the class of all temporal graphs which have the path $P_n$ as the underlying static graph}. 

\subsection{NP-Completeness} \label{sec:pathhard}

Observe that \ETC is in \NP because, for any input instance $(\GG,\Delta_1,\Delta_2,k)$, a non-deterministic algorithm can guess (if one exists) {a set $\Pi \subseteq \binom{V}{2}\times [T(\mathcal{G})]$ with $|\Pi|\leq k$} and, using Lemma~\ref{prop:verify-poly}, verify that {the temporal graph on $V$ with time-edge set $\mathcal{E}(\mathcal{G})\triangle \Pi$}  is a $(\Delta_1,\Delta_2)$-cluster temporal graph in time polynomial in $k$ and the size of $\mathcal{G}$. We show that \ETC is \NP-hard even for temporal graphs with a path as underlying graph.

\begin{theorem} \label{main_thm}
\EditTempClique is \NP-complete, even if the underlying graph $G$ of the input temporal graph $\GG$ is a path.
\end{theorem}

 For the remainder of this section we will consider only temporal graphs $\PP_n \in \mathfrak{F}_n$ which have a path $P_n$ as the underlying graph unless it is specified otherwise.

To prove this result we are going to construct a reduction to \ETC from the \NP-complete problem \TempMatch \cite{maximum_matchings}. For a fixed  $\Delta\in \mathbb{Z}^+$, a {\it $\Delta$-temporal matching} $\mathcal M$ of a temporal graph $\mathcal{G}$ is a set of time-edges {$\mathcal M \subseteq \mathcal E(\mathcal G)$} which are pairwise $\Delta$-independent. It is easy to note that if {$\mathcal E(\mathcal G)=\mathcal M$}, then $\mathcal G$ is a $(\Delta_1,\Delta)$-cluster temporal graph for any value of $\Delta_1 \ge 1$, because then each time-edge in $\mathcal G$ can be considered as a $\Delta_1$-temporal clique with unit time interval, and these cliques are by definition $\Delta$-independent. We can now state the \TempMatch problem formally. 
\begin{framed}
	\noindent
	\TempMatch (\TM):\\ 
 	\textit{Input:} A temporal graph $\mathcal G=(G,\TT)$ and two positive integers $k, \Delta \in \mathbb{Z}^+$.\\
\textit{Question:} Does $\mathcal G$ admit a $\Delta$-temporal matching $\mathcal M$ of size {at least} $k$?
\end{framed} 

{The main idea of the reduction in Theorem \ref{main_thm} is to add empty $3$ empty `filler' time slots between the time slots of our input graph (instance of \TM) and query the edit distance to a $(1,5)$-temporal cluster graph. The point of adding these time slots is to make it very costly to form cliques by adding new time-edges between non-independent edges of the input. Thus, an optimal solution to the \ETC problem will be to only remove time-edges, this will then recover a temporal matching of the desired size (should one exist). }

\begin{proof}[Proof of Theorem \ref{main_thm}] 
It was shown in \cite{maximum_matchings} that \TempMatch is \NP-complete for temporal graphs with a path as underlying graph and $\Delta=2$, thus we can assume our input has this form. Thus, let $\II=(\PP_n,2,k)$ denote an input instance of \TM, where $\mathcal P_n=(P_n,\TT) \in \mathfrak{F}_n$  and $k \in \mathbb{N}$.
From this, we are going to define an instance $\II'=(\PP'_n,1,5,k')$ of \ETC, which is a yes-instance if and only if $(\PP_n,2,k)$ is a yes-instance for \TM.  Here $\PP'_n=(P_n,\TT')\in \mathfrak{F}_n$ is a new temporal graph which has the same path $P_n$ as $\PP_n$ as the underlying graph, and we set $k':=|\mathcal{E}(\PP_n)| - k$.  A key property of our construction is that any pair of consecutive snapshots of $\PP_n$ is separated in $\PP'_{n}$ by three empty snapshots which we refer to as \textit{filler} snapshots; we will refer to all the other snapshots in $\mathcal P'_n$, which occur at the times $4t-3$ where $t\in \mathbb{Z^+}$, as \textit{non-filler} snapshots. Formally, our new temporal graph $\PP_n'$ has the same vertex set as $\PP_n$, and its set of time-edges is
\[\mathcal{E}(\mathcal{P}_n') = \left\{ (e, 4t - 3) : (e,t) \in \mathcal{E}(\PP_n)  \right\}.\]
We can assume that the first time-edge of $\PP_n$ occurs at time $1$ and so this also holds for $\PP_n'$. The construction of $\PP_n'$ from $\PP_n$ is illustrated in Figure~\ref{fig:resultslinething}. {We begin with the reverse direction of the reduction as this is simpler.} 
 
\begin{figure}
	\begin{center}
		\begin{tikzpicture}[xscale=.75,yscale=.75]
		
		 	\draw (-3.5,4.2) node[anchor=south]{{\large $\mathcal{P}_n$}};
		\draw[-] (-5,0) to node[pos=1.02,above]{} (-2,0);
		\foreach \x in {-5,...,-2}
		\draw (\x,0.1) to node[pos=1,below]{}   (\x,-0.1);
				\draw (-4.5,0) node[anchor=south]{\begin{tcolorbox}[ width=.75cm,height=3cm, halign=center,valign=center]
			\mbox{} 
			\end{tcolorbox}};
		
		\draw (-3.5,0) node[anchor=south]{\begin{tcolorbox}[ width=.75cm,height=3cm, halign=center,valign=center]
			\mbox{} 
			\end{tcolorbox}};
				
		\draw (-2.5,0) node[anchor=south]{\begin{tcolorbox}[ width=.75cm,height=3cm, halign=center,valign=center]
			\mbox{} 
			\end{tcolorbox}};
		
			\draw[fill] (-4.5,.5) circle (.1);
		\draw[fill] (-4.5,1.5) circle (.1);
		\draw[fill] (-4.5,2.5) circle (.1);
		\draw[fill] (-4.5,3.5) circle (.1);
		\draw (-4.5,2.5)[thick] -- (-4.5,3.5);
		
		\draw[fill] (-3.5,.5) circle (.1);
		\draw[fill] (-3.5,1.5) circle (.1);
		\draw[fill] (-3.5,2.5) circle (.1);
		\draw[fill] (-3.5,3.5) circle (.1);
		\draw (-3.5,0.5)[thick] -- (-3.5,1.5);
		\draw (-3.5,1.5)[thick] -- (-3.5,2.5);
		\draw (-3.5,2.5)[thick] -- (-3.5,3.5);
		
		\draw[fill] (-2.5,.5) circle (.1);
		\draw[fill] (-2.5,1.5) circle (.1);
		\draw[fill] (-2.5,2.5) circle (.1);
		\draw[fill] (-2.5,3.5) circle (.1);
		\draw (-2.5,2.5)[thick] -- (-2.5,3.5);

            	\draw (-1.5,2) [-To, line width=.6mm]-- (0.5,2);

			 	\draw (5.5,4.2) node[anchor=south]{{\large $\mathcal{P}_n'$}};
		
		\draw[-] (1,0) to node[pos=1.02,above]{} (10,0);
		\foreach \x in {1,...,10}
		\draw (\x,0.1) to node[pos=1,below]{}   (\x,-0.1);
		
		\foreach \x in {2.5,3.5,4.5,6.5,7.5,8.5}
		\draw (\x,0) node[anchor=south]{\begin{tcolorbox}[ width=.75cm,height=3cm, halign=center,valign=center,colback=black!25]
			\mbox{} 
			\end{tcolorbox}};

			\draw (1.5,0) node[anchor=south]{\begin{tcolorbox}[ width=.75cm,height=3cm, halign=center,valign=center]
			\mbox{} 
			\end{tcolorbox}};
		
			\draw (5.5,0) node[anchor=south]{\begin{tcolorbox}[ width=.75cm,height=3cm, halign=center,valign=center]
			\mbox{} 
			\end{tcolorbox}};

		\draw (9.5,0) node[anchor=south]{\begin{tcolorbox}[ width=.75cm,height=3cm, halign=center,valign=center]
			\mbox{} 
			\end{tcolorbox}};
	
	\draw[fill] (1.5,.5) circle (.1);
		\draw[fill] (1.5,1.5) circle (.1);
			\draw[fill] (1.5,2.5) circle (.1);
				\draw[fill] (1.5,3.5) circle (.1);
				 	\draw (1.5,2.5)[thick] -- (1.5,3.5);
				
					\draw[fill] (5.5,.5) circle (.1);
				\draw[fill] (5.5,1.5) circle (.1);
				\draw[fill] (5.5,2.5) circle (.1);
				\draw[fill] (5.5,3.5) circle (.1);
				\draw (5.5,0.5)[thick] -- (5.5,1.5);
				\draw (5.5,1.5)[thick,dashed] -- (5.5,2.5);
				\draw (5.5,2.5)[thick,dashed] -- (5.5,3.5);
				
					\draw[fill] (9.5,.5) circle (.1);
				\draw[fill] (9.5,1.5) circle (.1);
				\draw[fill] (9.5,2.5) circle (.1);
				\draw[fill] (9.5,3.5) circle (.1);
	\draw (9.5,2.5)[thick] -- (9.5,3.5);
\end{tikzpicture}
\end{center}	
\caption{\label{fig:resultslinething}The instance $\mathcal{P}$ to the temporal matching problem is shown on the left and the stretched graph $\mathcal{P}_n'$ on which we solve \probname $\,$ is on the right. Non-filler snapshots are shown in white and filler snapshots are grey.  Dotted edges show edges that were removed to leave a $(\Delta_1,\Delta_2)$-cluster temporal graph (which is also a temporal matching $\mathcal{M}'$).}
\end{figure}
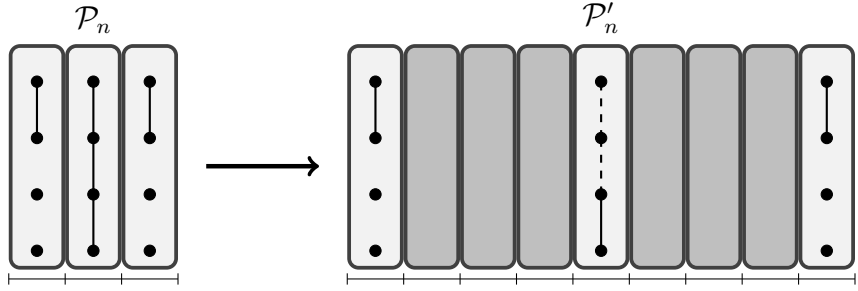

\bigskip

\noindent$[{\impliedby}]$ Let us suppose $(\mathcal{P}_n,2,k)$ is a YES-instance for \TM. This means that there exists at least one $2$-temporal matching $\mathcal{M}$ of size $k$ in $\mathcal{P}_n$. As we noticed before, $\mathcal M$ can be seen as a set of  pairwise $2$-independent $1$-temporal cliques each of which consists of a single time-edge.

\begin{clm}\label{clm:match} For any pair $(e,t),(e',t') \in \mathcal M$ such that $e \cap e' \not = \emptyset$, we have that $(e,4t - 3)$ and $(e',4t'-3)$ are $5$-independent.
\end{clm}
\begin{poc}[\ref{clm:match}]Because $(e,t)$ and $(e',t')$ both belong to  $\mathcal M$, they must be $2$-independent; since $e$ and ${e'}$ share at least one endpoint, it follows that $|t-t'| \geq 2$.  We therefore deduce that
\begin{equation*}
|(4t'-3) - (4t-3)| = 4|t'-t| \geq 8,
\end{equation*}
thus, the time-edges $(e,4t - 3)$ and $(e',4t'-3)$ are $5$-independent.\end{poc}

Set $\Pi:=\{(e,4t-3) \in \mathcal E(\PP'_n) \colon (e,t) \in \EE(\PP_n) \setminus \MM\}$ to be the set of time-edges to be removed from $\mathcal{P}_n'$. By construction $|\Pi|=k'$.  Thus, after deleting at most $k'$ time-edges from the temporal graph $\mathcal P'_n$ we obtain the matching $\mathcal{M}' = \{(e,4t-3) \colon (e,t) \in \mathcal{M}\}$, which by Claim \ref{clm:match} is a $(1,5)$-cluster temporal graph.  
 
\bigskip

\noindent$[ {\implies} ]$
For the {forward} direction, suppose $(\mathcal{P}'_n, 1, 5, k')$ with $\mathcal{P}'_n=(P_n,\TT') \in \mathfrak{F}_n$  is a YES-instance of \ETC. To begin we prove three claims about the structure of cliques in optimal solutions to \ETC on $\mathcal{P}_n'$.

\begin{clm}\label{clm:nonfillend}Let $\mathcal{K}$ be any $1$-temporal clique in an optimal solution to \ETC on $\mathcal{P}'_n$. Then the {endpoint}s of $L(\mathcal{K})$ must be non-filler snapshots of $\mathcal{P}_n'$.
\end{clm}
\begin{poc}[\ref{clm:nonfillend}] Assume for a contradiction that at least one endpoint of $L(\mathcal{K})$ is a filler snapshot. Filler snapshots do not contain any edges of $\mathcal{P}'_n$. Thus, when constructing $\mathcal{K}$, not adding the time-edges at this {endpoint} filler snapshot results in a $\Delta_1$-temporal clique $\mathcal{K}'$ with strictly fewer time-edges than $\mathcal{K}$, whose lifetime is strictly contained within the lifetime of $\mathcal{K}$.  Replacing $\mathcal{K}$ with $\mathcal{K}'$ clearly gives a temporal graph at lesser edit distance from $\PP_n'$ than $\mathcal{K}$, so it remains only to verify that this new temporal graph is in fact a $(1,5)$-cluster temporal graph.  To see this, note that if $\mathcal{K}$ is $5$-independent from all other cliques then so is any subset of $\mathcal{K}$, so we cannot violate the independence condition. We therefore obtain a contradiction to the hypothesis that $\mathcal{K}$ is part of an optimal solution. 
\end{poc}

\begin{clm}\label{clm:splitclique} Let $\mathcal{K}$ be any $1$-temporal clique on vertex set $X$ in a solution to \ETC on $\mathcal{P}'_n$ that contains two time-edges in distinct non-filler snapshots. Then there exists a non-empty set $\mathcal{D}$ of pairwise $5$-independent $1$-temporal cliques on the vertex set $X$ such that replacing $\mathcal{K}$ with $\mathcal{D}$ results in a $(1,5)$-cluster temporal graph at strictly lower total edit cost from $\mathcal{P}'_n$. Furthermore, for all $\mathcal{D}' \in \mathcal{D}$, we have  $L(\mathcal{D}') \subset L(\mathcal{K})$, and there exists a $\mathcal{D}' \in \mathcal{D}$ such that $L(\mathcal{D}')=\{t\}$, where $t$ is a non-filler snapshot of $\mathcal{P}'_n$.
 \end{clm}

\begin{poc}[\ref{clm:splitclique}] Let $\tau=\{t_1, \dots\} $ be the set of non-filler snapshots of $ \mathcal{P}'_n$. Let \begin{align*} \ell
	&= \min_{i\in \mathbb{Z}^+ }\{i: \mathcal{K} \text{ contains an edge at time } t_{i}\in \tau\}\\
	r &=\max_{i\in \mathbb{Z}^+} \{i: \mathcal{K} \text{ contains an edge at time } t_{i}\in \tau\},\end{align*} be the leftmost and rightmost non-filler snapshots containing a time-edge in $\mathcal{K}$, respectively. Note that, because $\Delta_1=1$ and $\mathcal{K}$ contains at least two time-edges which correspond to distinct snapshots in $\mathcal P_n$, it must be that $\mathcal{K}$ consists of a copy of the static clique on $X$ at every snapshot in its lifetime. Let $|X|=m$. Observe that  $\mathcal{K}$ has lifetime $[t_{\ell},t_{r}]$, by Claim \ref{clm:nonfillend}, and so by hypothesis there must be a copy of a clique on $X$ at least in the interval $[t_{\ell},t_{\ell+1}]$. Let $\mathcal{K}_1$ be the $1$-temporal clique on $X$ with lifetime $\{t_{\ell}\}$ and if $t_{r}>t_{\ell +1}$ then let $\mathcal{K}_2$ be the $1$-temporal clique on $X$ with lifetime $[t_{\ell+2},t_{r}]$, otherwise $\mathcal{K}_2$ is not defined. We will now prove that replacing $\mathcal{K}$ with the set of cliques $\mathcal{D}$ where $\mathcal{D} = \{\mathcal{K}_1,\mathcal{K}_2\}$ if $\mathcal{K}_2$ is defined, and $\{\mathcal{K}_1\}$ otherwise, will result in a $(1,5)$-cluster temporal graph with smaller edit distance from $\mathcal{P}'_n$.   
	
	To begin observe that, since $\mathcal{K}$ contains time-edges at least in $t_{\ell}$ and $t_{\ell+1}$, a copy of the clique on $X$ must be added to each snapshot in $[t_{\ell}+1,t_{\ell+1}-1]$ to form $\mathcal{K}$. Since these edges are not present in $\mathcal{K}_1$ or $\mathcal{K}_2$, not adding them saves $3\binom{m}{2}$ edge additions. Note that if $t_{r}>t_{\ell +1}$ then additional edge additions in the interval $[t_{\ell+1}+1,t_{\ell+2}-1]$ will be saved. When forming $\mathcal{K}_1$ we must remove all time-edges with both endpoints in $X$ from the snapshot $t_{\ell+1}$, since the underlying graph is a path there are at most $m-1$ such edges. At all snapshots other than $t_{\ell}+1,\dots, t_{\ell+2}-1$ we make the same changes as when forming $\mathcal{K}$ and thus there is no net change in the number of edits. Thus, at least 
	\[3\cdot \binom{m}{2} -(m-1) >0,\]fewer edits are required to form $\mathcal{K}_1$ (and $\mathcal{K}_2$ if $t_{r}>t_{\ell +1}$) over forming $\mathcal{K}$. 
	
	Finally, we can note that even if $\mathcal{K}_2$ was non-empty, then $\mathcal{K}_1$ and $\mathcal{K}_2$ would be still $5$-independent by construction. Moreover, since the time-edges of both $\mathcal{K}_1$ and (if it is defined) $\mathcal{K}_2$ are subsets of the time-edges of $\mathcal{K}$, both cliques must be $5$-independent from any other temporal clique in the original solution.\end{poc}
	
	Using the two claims above we can establish the following claim.

\begin{clm}\label{clm:oneshot} Each clique $\mathcal{K}$ in any optimal solution to \ETC on $\mathcal{P}'_n$ satisfies $L(\mathcal{K})=\{t\}$, where $t$ is a non-filler snapshot of $\mathcal{P}'_n$.\end{clm}
\begin{poc}[ \ref{clm:oneshot}]Assume for a contradiction that some $1$-temporal clique $\mathcal{K}$ in an optimal solution to \ETC has a lifetime greater than one. By Claim \ref{clm:nonfillend} we can deduce that $L(\mathcal{K})$ contains at least two consecutive non-filler snapshots. It follows immediately from Claim~\ref{clm:splitclique} that our solution is not optimal, since we know that we can obtain another solution with lesser total edit cost by replacing $\mathcal{K}$ with a different set $\mathcal{D}$ of temporal cliques, each of which has a smaller lifetime than $\mathcal{K}$. This gives the required contradiction.\end{poc}

The next claim is key to the reduction.

\begin{clm}\label{clm:oneedge} There exists an optimal solution to \ETC on $\mathcal{P}'_n$ where each $1$-temporal clique consists of a single time-edge $(e,t)$, where $t$ is a non-filler snapshot of $\mathcal{P}'_n$. Furthermore, this solution can be obtained using deletions only.\end{clm}

\begin{poc}[ \ref{clm:oneedge}]Let $\{\mathcal{K}_i: 1 \le i \le s\}$ be an optimal solution to our instance of \ETC.  First of all let us note that, by Claim \ref{clm:oneshot}, each $1$-temporal clique $\mathcal{K}_i$ consists of a collection of time-edges that all appear at a single non-filler snapshot $t$ of $\mathcal{P}'_n$. We are going to prove that any of these $1$-temporal cliques can be replaced with a set of pairwise $5$-independent single time-edges, all of which are present in the input temporal graph, without increasing the total edit cost.

Let us fix a clique $\mathcal{K}_i$ in our optimal solution, and suppose that $L(\mathcal{K}_i) = \{t\}$, where $t$ is a non-filler snapshot of $\mathcal{P}'_n$.  Since the underlying graph is a path, the set of edges appearing at time $t$ consists of either a single path or a disjoint collection of paths.  By Lemma \ref{induced_connectivity}, we may assume that the time-edges of $\mathcal{K}_i$ which also belong to $\mathcal{P}_n'$ form a single path.  Suppose that this path has $\ell$ vertices.

First of all we notice that if $\ell =2$, then it is already a $1$-temporal clique with one edge, so there is no modification to apply.

If $\ell=3$, then $P_3$ consists of two time-edges $\{(uv,t),(vw,t)\}$, so either deleting one of these time-edges, or adding the missing time-edge $(uw,t)$ would create a $1$-temporal clique. Since the cost of doing either is the same, we can assume, without loss of generality, that an edge is deleted to obtain a $1$-temporal clique consisting of a single time-edge.

Let us now assume $\ell>3$. Then, deleting $\lfloor \frac{\ell}{2} \rfloor$ time-edges we would obtain a set of $1$-temporal cliques which consists of single time-edges and which are pairwise disjoint in the underlying graph, and so are pairwise $5$-independent. We now consider the cost of constructing a $1$-temporal clique from $P_{\ell}$ that contains all of its vertices. In this case we will need to add exactly $\frac{\ell(\ell-1)}{2} - (\ell-1)=\frac{1}{2}(\ell-1)(\ell-2)$ time-edges in $\mathcal{P}'_n$ at time $t$. Thus since $\ell>3$ the number of edges saved by using deletions only is
\[\frac{(\ell-1)(\ell-2)}{2}- \left\lfloor \frac{\ell}{2}\right\rfloor {\geq  \ell-1  - \frac{\ell }{2}  > 0}.\]   

Hence, for any $\ell\geq 1$ there exists a solution with minimum cost which consists of pairwise $5$-independent time-edges and can be obtained using only deletion.\end{poc}

By Claim \ref{clm:oneedge} above, there exists an optimal solution to \ETC on $\mathcal{P}'_n$ consisting of pairwise $5$-independent time-edges that is obtained using deletions only. Since $(\mathcal{P}_n',1,5,k')$ is a yes-instance, we know that there exists such a solution in which the set $\Pi$ of time-edges deleted satisfies $|\Pi| \leq k'$.  
Let $\mathcal{M}'$ be the set of time-edges in the resulting temporal graph, and note that every element of $\mathcal{M}'$ appears in a non-filler snapshot.

\begin{clm}\label{clm:MisMatch}$\mathcal{M} := \{(e,(t+3)/4) \colon (e,t) \in \mathcal{M}'\}$ is a $2$-temporal matching of $\mathcal{P}_n$.
\end{clm}
\begin{poc}[ \ref{clm:MisMatch}]
It suffices to demonstrate that, for every pair of time-edges $(e_1,t_1)$ and $(e_2,t_2)$ in $\mathcal{M}$, either $e_1 \cap e_2 = \emptyset$ or $|t_2 - t_1| \ge 2$.  Suppose that $e_1 \cap e_2 \neq \emptyset$.  By construction of $\mathcal{M}$, we know that $(e_1,4t_1 -3)$ and $(e_2,4t_2 - 3)$ belong to $\mathcal{M}'$.  Since $\mathcal{M}'$ is a set of single time-edges which form a $(1,5)$-cluster temporal graph, we know that the temporal cliques consisting of the single time-edges $(e_1,4t_1 - 3)$ and $(e_2,4t_2 - 3)$ are $5$-independent; as we are assuming that $e_1 \cap e_2 \neq \emptyset$, this implies that $|(4t_2 - 3) - (4t_1 - 3)| \ge 5$.  We can therefore deduce that $4|t_2 - t_1| \ge 5$, implying that $|t_2 - t_1| > 1$; since $t_1,t_2 \in \mathbb{Z}^+$, it follows that $|t_2 - t_1| \ge 2$, as required.\end{poc}

 Thus by Claim \ref{clm:MisMatch} if $(\mathcal{P}'_n, 1,5, k')$ is a YES-instance of the \ETC problem with $k'=|\mathcal{E}(\mathcal{P})| -k$ then $(\mathcal{P}'_n, 2, k)$ is a YES-instance of the \TM problem since $|\mathcal{M}| \geq |\mathcal{E}(\mathcal{P})| -k' = k$. \end{proof}

\subsection{Bounding the Number of Edge Appearances}\label{sec:pathspoly}

  We have seen from the previous section that \ETC is \NP-hard, even on temporal graphs $\mathcal{P}_n \in \mathfrak{F}_n$ which have a path as their underlying graph. In this section we show that, if additionally the number of appearances of each edge in $\mathcal{P}_n$ is bounded by a fixed constant, then the problem can be solved in time polynomial in the size of the input temporal graph.   
\begin{theorem}\label{thm:DPonPath}
Let $(\mathcal P_n, \Delta_1, \Delta_2,k)$ be any instance of \ETC where $\mathcal{P}_n \in \mathfrak{F}_n$ and there exists a natural number $\sigma$ such that $|\mathcal T(e)| \leq \sigma$ for {every} $e\in E(P_n)$. Then, \ETC on $(\mathcal P_n, \Delta_1, \Delta_2,k)$ is solvable in time $\mathcal{O}(T^{4\sigma} \sigma^2\cdot n^{2\sigma +1 })$.
\end{theorem}

We prove this theorem using dynamic programming where we go along the underlying path $P_n$ uncovering one vertex in each step. In particular, if we have reached vertex $i$ then we try to extend the current set of templates on the first $i-1$ vertices of the path to an optimal set of templates also including the $i$\textsuperscript{th} vertex.

Let $\mathcal{C}^i=(C_1, \dots, C_r)$ be a (potentially empty) collection of templates, where each template $C_j$ with $1\leq j \leq r$ has vertex set $X_j\subseteq \{v_1, \dots, v_i\} $ {satisfying} $|X_j| \geq 2$. For any vertex $v_i \in V(P_n)$ with $i \geq 2$, we will say that $\mathcal{C}^i$ is {\it feasible for $v_i$} if it is pairwise $\Delta_2$-independent and respects the following:

\begin{enumerate}[(i)]
    \item\label{itm:feas1} for each $C_j = (X_j,[a_j,b_j])$, we have $X_j = \{v_{\ell_j}, v_{\ell_j+1},\dots,v_i\}$ for some $\ell_j< i$,
   \item\label{itm:feas2} for each $j$, for any $k \in \{\ell_j,\dots,i-1\}$ there exists $t \in [a_j,b_j]$ such that $t \in \mathcal T(v_kv_{k+1})$.
\end{enumerate}
For the case $i=1$ we define the empty collection $\mathcal{C}=\emptyset$ to be the only collection feasible for $v_1$. Let $\Gamma^i$ be the class of collections of templates feasible for  $v_i$. {Feasible templates will be the key to our dynamic program. In particular, they form the interface between the part of the solution that has already been fixed, and what can still be edited in the future steps.}

\begin{lemma}\label{explaincor}
Let $n\in \mathbb{Z}^+$, $\mathcal{P}_n \in \mathfrak{F}_n$, and $\mathcal C \in \mathfrak{T}({\mathcal{P}_n},\Delta_2)$  be a collection minimising \[\min_{\mathcal{G}_{\mathcal C} \text{ $\Delta_1$-realises } \mathcal C} |{\mathcal{P}_n} \triangle \mathcal{G}_{\mathcal C}|.\] Then, for any set $\mathcal{C}^n\subseteq \mathcal C$ of templates that contains $v_n$  we have $\mathcal{C}^n\in\Gamma^n$.   
\end{lemma}

\begin{proof}
Given an arbitrary non-empty template $C_j=(X_j,[a_j,b_j])\in \mathcal C^n$ we shall show that it satisfies Properties \eqref{itm:feas1} and \eqref{itm:feas2} in the definition of feasibility for $v_n$.

 \textit{Property \eqref{itm:feas1}}:
 By construction, because $v_n \in X_j$ is an end-vertex of the underlying path $P_n$, and $|X_j| \geq 2$, we can assume that the vertex set $X_j$ of $C_j$ contains a non-empty sequence of consecutive vertices in $P_n$, that is $\{v_{l_j},v_{{l_j}+1}, \dots, v_n\} \subseteq X_j$ and $v_{{l_j}-1} \not \in X_j$  for some $l_j \leq n$. We now claim that $X_j \subseteq \{v_{l_j},v_{{l_j}+1}, \dots, v_n\}$. To see this, note that if there exists a further vertex $v \in X_j \setminus \{v_{l_j},v_{{l_j}+1}, \dots, v_n\}$ such that $v \neq v_{{l_j}-1}$, then $P_n[C_j]$ would be disconnected, which contradicts Lemma \ref{induced_connectivity}.

 \textit{Property \eqref{itm:feas2}}: Now, as we have shown that $C_j$ satisfies Property \eqref{itm:feas2}, we can assume $X_j=\{v_{l_j}, \dots, v_n\}$. If we assume that there exists a $k \in \{l_j,\dots,n-1\}$ such that $[a_j,b_j]\cap \mathcal{T}(v_{k}v_{k+1})= \emptyset$ then $C_j$ does not induce a static  connected sub-graph in $\mathcal{P}_n$. This is in contradiction with Lemma \ref{induced_connectivity}.

The result then holds for any $C_j\in \mathcal C$ since they were arbitrary.
  \end{proof}

Let $\mathcal C^{i-1}=\{C^{i-1}_1,{\dots},C^{i-1}_d\}$ and $\mathcal{C}^{i}=\{C^i_1,{\dots},C^i_r\}$ be collections of templates which are feasible for $v_{i-1}$ and $v_i$ in $V(P_n)$ respectively. Then, we say that $\mathcal C^i$ $extends$ $\mathcal C^{i-1}$ if for all $C^{i}_k\in \mathcal{C}^{i}$ where $|C^{i}_k|>2$, there exists a $j \in \{1,{\dots},d\}$ such that $V(C^{i-1}_j) \cup \{v_i \} = V(C^{i}_k)$ and $L(C^{i-1}_j)  = L(C^{i}_k)$. We denote this by $\mathcal C^{i-1} \prec \mathcal C^i$. The following Lemma is useful in the proof of Theorem \ref{thm:DPonPath}.

\begin{lemma} \label{extention}
 Let $n\in \mathbb{Z}^+$, $\mathcal{P}_n \in \mathfrak{F}_n$, and $\mathcal C^i$ be any collection of feasible  templates for $v_i\in V(P_n)$. Then, $\mathcal C^{i-1} = \left\{ (X\backslash \{v_i\}, [a,b]) : (X, [a,b])\in \mathcal{C}^i, \;|X|>2 \right\},$ is feasible for $v_{i-1}\in V(P_n)$ and $\mathcal C^{i-1} \prec \mathcal C^{i} $.
\end{lemma}

\begin{proof}
We shall prove this directly by checking feasibility while constructing $\mathcal{C}^{i-1}$ from $\mathcal{C}^{i}$. Let $C_j=(X_j,I_j)$ be any template in $\mathcal C^i$. If $X_j=\{v_{i-1},v_i\}$, then there is nothing to prove and we set $C'_j=\emptyset$. Let us assume that $|X_j| > 2$. By the definition of feasibility $X_j=\{v_{\ell_j}, \dots, v_i \}$ for some $\ell_j < i$. Let us consider the new template $C'_j$ given by $C'_j=(X_j \setminus \{v_i\}, I_j)$. Then,  $C'_j$  satisfies {Property \eqref{itm:feas1}} of feasibility for $v_{i-1} \in V(P_n)$. Furthermore, because $C_j$ satisfies {Property \eqref{itm:feas2}}  for $v_i$, then by construction also $C'_j$ satisfies that property for $v_{i-1}$. Thus, given $\mathcal{C}^i$, if we let $\mathcal{C}^{i-1} = \{C'_j:C_j \in \mathcal C^i\}$, then we have $\mathcal C^{i-1} \prec \mathcal C^{i} $. 
  \end{proof}

 We can now state the dynamic program. The idea is to revel the next vertex $v_i$ in the temporal path at each step and extend a set of templates $\mathcal{C}^{i-1}$ feasible for $v_{i-1}$ to set of templates $\mathcal{C}^{i}$ feasible for $v_i$. The function $\vartheta_i[\mathcal{C}^{i}]$ keeps track of the minimum number of edges that must be added and removed along the way to make the temporal graph realise an optimal set of templates such that the subset of templates containing $v_i$ is equal to $\mathcal{C}^{i}$.

Let $i\geq 2$ and $\mathcal{C}^i\in \Gamma^i$ be any collection feasible for $v_i$, where $C_j^i=(X_j^i,[a_j,b_j])$  for each $C_j^i \in \mathcal{C}^i$. Then, using the functions $f_i,g_i$ and $h_i$ given by \eqref{eq:gfunction}, \eqref{eq:hfunction} and \eqref{eq:ffunction} below, we define the function $\vartheta_i[\mathcal C^i]$ as follows: \begin{equation}\label{eq:pathrecursion}\vartheta_i[\mathcal C^i]=\min\limits_{\mathcal C^{i-1} \prec \mathcal C^i} \left[\vartheta_{i-1} [ \mathcal C^{i-1}] + f_i(\mathcal C^{i}) +  g_i( \mathcal C^i) + h_i(\mathcal{C}^i)\right],\end{equation}  and $\vartheta_{1} [\mathcal C^1]=0$, where $\mathcal{C}^1=\emptyset$ is the only collection feasible for $v_1$. We now state the functions $h_i,f_i,g_i$.
\begin{itemize}
	\item The function $f_i: \Gamma^i \rightarrow \mathbb N$ counts the {minimum} total number of time-edges from $v_i$ to vertices $u\neq v_{i-1}$ which appear in the $\Delta_1$-temporal cliques which realise all the $C_j^{i} \in \mathcal{C}^{i}$ with more than two vertices. This function
	is given by \begin{equation}\label{eq:gfunction}f_i(\mathcal C^{i})=\sum_{j=1}^{|\mathcal {C}^i|} \left(| X^{i}_j|-2\right)\cdot \max\left(\left\lfloor \frac{|b_j-a_j|}{\Delta_1} \right\rfloor, 1\right).\end{equation}
	
	\item The function $g_i: \Gamma^i \rightarrow \mathbb N$ counts the number of appearances of $v_{i-1}v_i$ that do not lie in any $\Delta_1$-temporal clique which $\Delta_1$-realises any template in $\mathcal{C}^{i}$. This function
	is given by   \begin{equation}\label{eq:hfunction}g_i(\mathcal C^i)= \left|\mathcal T(v_{i-1}v_i) \backslash \bigcup_{j=1}^{|\mathcal{C}^i|} [a_j,b_j]\right|.\end{equation}

		\item The function $h_i:   \Gamma^i \rightarrow \mathbb N$ counts the minimum total number of time-edges from $v_i$ to $v_{i-1}$ which must appear in the $\Delta_1$-temporal cliques realising each $C_j^{i} \in \mathcal{C}^{i}$. Firstly, for each  template $C^i_k=(X^i_k,[a_k,b_k])$ we define $\mathcal T^i_k=\mathcal{T}(v_{i-1}v_i)\cap [a_k,b_k]$. Let $t_0=a_k$, $t_{|\mathcal T^i_k|+1}=b_k$, and $(t_j)_{j=1}^{|\mathcal T^i_k|}$ be the elements of $\mathcal T^i_k$ ordered increasingly. Then,
	\begin{equation}\label{eq:ffunction}h_i( \mathcal{C}^i) = \sum_{k=1}^{|\mathcal {C}^i|}  \sum_{j=0}^{|\mathcal T^i_k|+1}   \left\lfloor \frac{t_{j+1}-t_j-1}{\Delta_1} \right\rfloor.\end{equation} 
\end{itemize}

The next lemma shows the recurrence given by \eqref{eq:pathrecursion} counts {the minimum number of edge edits to transform the given temporal graph on a path into a temporal cluster graph realising a given set of templates.}

\begin{lemma} \label{main}
		Let $i,n\in \mathbb{Z}^+$ where $i\leq n$, $\mathcal{P}_n \in \mathfrak{F}_n$, and $\mathcal C^i=\{C^i_1,{\dots},C^i_r\}$ be any collection of templates feasible for the vertex $v_i \in V(P_n)$. Let $\operatorname{Opt}(\mathcal{C}^i)$ be the minimum number of modifications required to transform $\mathcal P_i=\mathcal P_n[v_1,\dots,v_i]$ into a $(\Delta_1,\Delta_2)$-cluster temporal graph which $\Delta_1$-realises a collection  $\mathcal{Q}^i $ of pairwise $\Delta_2$-independent templates such that $\{Q \in \mathcal \mathcal{Q}^i: v_i \in Q\}=\mathcal C^i$. Then,
\[\vartheta_i[\mathcal C^i] = \operatorname{Opt}(\mathcal{C}^i).\]
\end{lemma}

\begin{proof} We will prove this equality by induction on  $i\leq n$, the number of vertices in the induced path. The following claim will be used to prove the inductive step.
\begin{clm}\label{clm:viedges} Let $\mathcal{C}^{i-1}\in \Gamma^{i-1}$,  $\mathcal{C}^i\in \Gamma^i$ and $\mathcal{C}^{i-1}\prec  \mathcal{C}^i$. Suppose we have modified $\mathcal{P}^{i-1}$ into a  $(\Delta_1,\Delta_2)$-cluster temporal graph realising a collection $\mathcal{Q}^{i-1} $ such that $\{Q \in \mathcal \mathcal{Q}^{i-1}: v_{i-1} \in Q\} = \mathcal{C}^{i-1}$.  Then, the minimum number of additions and deletions of time-edges of the form $(v_{x}v_i,t)$, where $x<i$, required to modify $\mathcal{P}^i$ into a $(\Delta_1,\Delta_2)$-cluster temporal graph realising a collection $\mathcal{Q}^i $ such that $\{Q \in \mathcal \mathcal{Q}^i: v_i \in Q\} = \mathcal{C}^i$ is   \[ f_i(\mathcal{C}^i)+ g_i(\mathcal{C}^i)+h_i(\mathcal{C}^i).\] 
\end{clm}

\begin{poc}[ \ref{clm:viedges}]
Note that at step $i-1$ of this algorithm we already have selected a collection $\mathcal C^{i-1} \in \Gamma^{i-1}$ such that $\mathcal C^{i-1} \prec \mathcal C^{i} $. 
So, in order to extend $\mathcal C^{i-1}$ to $\mathcal C^i$ we need to count the total number of modifications that we need to apply to any edge of the form $v_xv_i$ with $x<i$ in such a way that it results to be $\Delta_1$-dense exactly in the lifetimes of each template in $\mathcal C^i$. We need to count: (i) the minimum {number of} time-edges from $v_i$ to $v_{j}$, where $j<i-1$, we need to add to make each template a $\Delta_1$-temporal clique, (ii) time-edge from $v_i$ to $v_{i-1}$ outside the lifetime of any template that we must remove, and finally (iii) the minimum time-edges from $v_i$ to $v_{i-1}$ we need to add to make each template a $\Delta_1$-temporal clique. Since our input temporal graph has a path as the underlying graph, this covers all possible time-edges of the form $(v_iv_x,t)$.

To begin we must count, for each $C^i_j \in \mathcal C^i$, the minimum number of time-edges $(v_zv_i,t)$ where $v_z \in   X^i_j \setminus \{v_{i-1},v_{i}\}$ and $t \in [a_j,b_j]$, that we need to add in order to link the new vertex $v_i$ to all the other vertices in $v \in X^i_j \setminus \{v_{i-1},v_{i}\}$, in such a way every edge $vv_i$ is $\Delta_1$-dense in $[a_j,b_j]$. This number is given by $\left(| X^{i}_j|-2\right)\cdot \max\left(\left\lfloor \frac{|b_j-a_j|}{\Delta_1} \right\rfloor, 1\right)$. Summing over $j \in \{1,\dots,|\mathcal C^i|\}$, the minimum number of edges of the kind $v_z v_i$,  where $z\leq i-2$, we need to add is given by

\[f_i(\mathcal C^{i})=\sum_{j=1}^{|\mathcal { C}^i|} \left(| X^{i}_j|-2\right)\cdot \max\left(\left\lfloor \frac{|b_j-a_j|}{\Delta_1} \right\rfloor, 1\right).\]

Then, one must consider the number appearances of the edge $v_{i-1}v_i$ which lie outside the lifetime of any template $C_{j}^i=(X_j,[a_j,b_j])$ contained in $\mathcal C^i$, and thus must be deleted. This value is accounted by the function

\[g_i( \mathcal C^i)=\left|\mathcal T(v_{i-1}v_i) \setminus \bigcup_{j=1}^{|\mathcal{C}^i|} [a_j,b_j]\right|.\] 

Finally we should calculate the minimum number of appearances of the edge $v_{i-1}v_i$ that should be added to $\mathcal T(v_{i-1}v_i)$ so that $v_{i-1}v_i$ results to be $\Delta_1$-dense within each template. To do so, let us consider any $C^i_k=(X_k,[a_k,b_k]) \in \mathcal {C}^i$ and let $\mathcal T_k^i=\mathcal T (v_{i-1}v_i) \cap [a_k,b_k]$ be the set of all the time appearances of the edge $v_{i-1}v_i$  in the interval $[a_k,b_k]$. Now, let $(t_1,{\dots},t_{|\mathcal T_k^i|})$ be the ordered list of all the time appearances of $v_{i-1}v_i$ in $[a_k,b_k]$, and denote $ t_0 = a_k$ and $t_{|\mathcal T_k^i|+1} = b_k$. Then, the sum \begin{equation}\label{fsum} \sum_{j=0}^{|\mathcal T_k^i|+1}   \left\lfloor \frac{t_{j+1}-t_j-1}{\Delta_1} \right\rfloor, \end{equation} 
counts exactly the minimum number of appearances of $v_{i-1}v_i$ that are missing in $[a_k,b_k]$ for any pair $(t_j,t_{j+1})$ of consecutive appearances of $v_{i-1}v_i$ such that $|t_{j+1}-t_j| > \Delta_1$. Summing over all the possible $C^i_k \in \mathcal {C}^i$ gives

\begin{equation*} h_i( \mathcal{C}^i) = \sum_{k=1}^{|\mathcal {C}^i|}  \sum_{j=0}^{|\mathcal T_k^i|+1}   \left\lfloor \frac{t_{j+1}-t_j-1}{\Delta_1} \right\rfloor.\end{equation*}

The claim now follows from summing $f_i(\mathcal{C}^i),g_i(\mathcal{C}^i)$ and $h_i(\mathcal{C}^i)$.\end{poc}

{Recall from the statement of Lemma \ref{main} that $\operatorname{Opt}(\mathcal{C}^i)$ is the minimum number of modifications required to transform $\mathcal P_i$ into a $(\Delta_1,\Delta_2)$-cluster temporal graph realising a collection  $\mathcal{Q}^i $ of pairwise $\Delta_2$-independent templates with $\{Q \in \mathcal \mathcal{Q}^i: v_i \in Q\}=\mathcal C^i$. We now prove $\vartheta_i[\mathcal C^i] = \operatorname{Opt}(\mathcal{C}^i)$ by induction, by showing non-strict inequalities hold in both directions for an arbitrary collection $\mathcal C^i$ of templates feasible for $v_i$. }

\medskip

We first prove $\vartheta_i[\mathcal C^i] \leq \operatorname{Opt}(\mathcal{C}^i)$\textbf{:} For the base case ($i=1$), by the definition of feasibility, the only collection $\mathcal{C}^1$ of templates feasible for $v_1$ is the empty one. As $\vartheta_{1} [\mathcal{C}^1]=0 $ and any temporal graph on one vertex is a $(\Delta_1,\Delta_2)$-cluster temporal graph the base case follows. Now, from  $\mathcal C^i$ we define the collection $\mathcal C^{i-1} = \left\{ (X\backslash \{v_i\}, [a,b]) : (X, [a,b])\in \mathcal{C}^i, \;|X|>2 \right\}$. Proposition \ref{extention} shows that $\mathcal{C}^{i-1}  \in \Gamma^{i-1}$, and $\mathcal{C}^{i-1}\prec \mathcal{C}^{i} $. {Let $S_{add}$ be any minimal set of time-edges which must be added to $\mathcal{P}_i$ in order for each template $C_i\in \mathcal{C}^i$ to be a $\Delta_1$-temporal clique, and $S_{rem}$ be any minimal set of time-edges which must be removed from $\mathcal{P}_i$ in order for each template $C_i\in \mathcal{C}^i$ to be $\Delta_2$-independent from all other edges of $\mathcal{P}_i$. The set $S_{add}\cup S_{rem}$ can be decomposed two disjoint sets: $S_1$ consisting of time-edges which have $v_i$ as an end-point, and  $S_0$ contains those that do not. Since $\mathcal C^{i-1} $ contains precisely the templates of $\mathcal C^{i}$ with three or more vertices, it follows that $\operatorname{Opt}(\mathcal{C}^i) = \operatorname{Opt}(\mathcal{C}^{i-1}) + |S_1| $. Now, by inductive hypothesis we have $\operatorname{Opt}(\mathcal{C}^{i-1}) \geq \vartheta_{i-1}[\mathcal C^{i-1}]$ and by Claim \ref{clm:viedges} we have $|S_1| =  f_i(\mathcal{C}^i)+ g_i(\mathcal{C}^i)+h_i(\mathcal{C}^i)$, hence $ \operatorname{Opt}(\mathcal{C}^i) \geq \vartheta_i[\mathcal C^i]$.} 
\medskip 

We now prove $\vartheta_i[\mathcal C^i] \geq \operatorname{Opt}(\mathcal{C}^i)$\textbf{:} {Since $\operatorname{Opt}(\mathcal{C}^i)$ considers all collections $\mathcal{Q}^i$ satisfying $\{Q \in \mathcal \mathcal{Q}^i: v_i \in Q\}=\mathcal C^i$, it remains to show that $\vartheta_i[\mathcal C^i] $ is a lower bound on the number of time-edge edits needed to transform $\mathcal{P}_i$ into a $(\Delta_1,\Delta_2)$-temporal cluster graph. For the base case ($i=1$), by the definition of feasibility, the only collection $\mathcal{C}^1$ of templates feasible for $v_1$ is the empty one, and $\vartheta_{1}(\mathcal{C}^{1})=0$. Assume by induction that for any $\mathcal{C}^{i-1} \in \Gamma^{i-1}$, the minimum number of time-edge edits required to make $\mathcal{P}_{i-1}$ realise a collection $\mathcal{Q}^{i-1}$ satisfying $\{Q \in \mathcal \mathcal{Q}^{i-1}: v_i \in Q\}=\mathcal C^{i-1}$ is bounded from below by $\vartheta_{i-1}(\mathcal{C}^{i-1})$. Since this bound holds for any $\mathcal{C}^{i-1}\prec \mathcal{C}^{i} $, the lower bound then follows by Claim \ref{clm:viedges} and the definition of $\vartheta_{i}(\mathcal{C}^{i})$.}  \end{proof}
We are now ready to prove Theorem \ref{thm:DPonPath}.

\begin{proof}[Proof of Theorem \ref{thm:DPonPath}] We will implement the function $\vartheta_i$, defined in \eqref{eq:pathrecursion}, as a dynamic program to solve the \ETC problem on $(\mathcal P_n, \Delta_1, \Delta_2,k)$. We start from $\vartheta_1[\mathcal C^1]=0$, where $\mathcal C^1=\emptyset$ by definition of feasibility, and tabulate bottom up until the value of $\vartheta_n[\mathcal C^n]$ is calculated for  any possible choice of $\mathcal C^n\in \Gamma^n$. The result then follows by Lemma \ref{explaincor} since we minimise $\vartheta_n[\mathcal C^n]$ over all choices of $\mathcal C^n\in \Gamma^n$.

 We assume that for any edge $e\in E(\mathcal P_i)$, the set $\mathcal T(e)$ of its time appearances is ordered. Since each edge appears at most $\sigma$ times this prepossessing step takes time at most  $\mathcal{O}(n\sigma\log \sigma)$ and so is negligible. We also assume that arithmetic operations on $\mathcal{O}(\log n)$-bit numbers can be performed in constant time. 
 
At each step $i$ the algorithm must scan over all possible choices for the collection $\mathcal C^i$ and, for each of them, calculate $\vartheta_i[\mathcal C^i]$. The following Claim bounds the number of possibilities for the choice of $\mathcal C^i$ from above. 
\begin{clm}\label{clm:countingfeasible}
For any $1\leq i\leq n$ we have $|\Gamma^i|\leq T^{2\sigma}n^\sigma$.
\end{clm}
\begin{poc}[\ref{clm:countingfeasible}]
Let $\mathcal C^i=\{C_1,\dots, C_r\}$ be any collection of  templates contained in $\Gamma^i$. By hypothesis, $|\mathcal T(v_{i-1}v_i)|\leq \sigma$, so, because $\mathcal G[C_j]$ with $j \in [r]$ is connected by Lemma \ref{induced_connectivity}, then the vertex $v_i$ can appear in at most $\sigma$ templates in $\mathcal C^i$, thus there are at most $\sigma$ templates in $\mathcal{C}^i$. From Lemma \ref{explaincor}, the vertex set of each template $C_j^i \in \mathcal{C}^i$ which contains at least two vertices, is of the form $X_j=\{v_l, \dots, v_i \}$ with $l \leq i$. Thus, each $X_j$ is completely specified by its left-hand end-vertex $v_\ell$.  There are at most $n^\sigma$ ways to choose the $\sigma$ left-hand endpoints. Finally, there are $\binom{T}{2}\leq T^{2}$ ways to choose the start and finish points of the lifetime interval $[a_j,b_j]$ of a given template $C_j$, thus there are at most $T^{2\sigma}$ ways to choose the time intervals. So, in total there are at most $T^{2\sigma}n^\sigma$ ways to choose the collection $\mathcal{C}^i$.
\end{poc}

Now, given a collection $\mathcal C^i$ of feasible templates for the vertex $v_i$ of the path, let us analyse the time needed to compute the functions $f_i$, $g_i$ and $h_i$. 

From the definition \eqref{eq:gfunction} of $f_i(\mathcal{C}^i)$ we can see $f_i$ is a sum of $|{\mathcal{C}}^i| $ items. Each of these items is of the form $(|X^i_j|-2)\cdot \max\left( \lfloor\frac{|b_j-a_j|}{\Delta_1}\rfloor , 1\right)$. Recall that $|\mathcal{C}^i| \leq \sigma$, so $f_i$ can be calculated in $\mathcal{O}(\sigma)$  time.

From \eqref{eq:hfunction}, {we can compute $g_i( \mathcal C^i)$ by checking} if each time appearance of $e_i$ lies within one of the intervals $[a_1, b_1]\dots, [a_{|\mathcal{C}^i|}, b_{|\mathcal{C}^i|}]$. This takes time at most $\mathcal{O}(\sigma^2)$ since each edge appears at most $\sigma$ times and $|\mathcal{C}^i|\leq \sigma$.

Finally, we see from \eqref{eq:ffunction} that $h_i( \mathcal{C}_i)$ is a double sum. The inner sum needs to compute the set $\mathcal{T}^i_k$ and order it increasingly. Since $\mathcal{T}(v_{i-1}v_i)$ is already sorted this takes time $\mathcal{O}(\sigma)$. Having computed this it performs $|\mathcal{T}^i_k|+2\leq \sigma+2$ many operations, each costing $\mathcal{O}(1)$ time. The outer sum has $|\mathcal{C}^i|\leq \sigma $ many items, and so we conclude that computing $h_i(\mathcal{C}_i)$ requires $\mathcal{O}(\sigma^2) $ time.

Since we are tabulating from the bottom up, the values $\vartheta_{i-1}[\mathcal C^{i-1}]$ for any $\mathcal C^{i-1}\in \Gamma^{i-1}$ are already known. So the computational cost required for looking up each such value is constant. For any two fixed collections $\mathcal{C}^{i-1}$ and $\mathcal{C}^{i}$, it takes time at most $\mathcal{O}(\sigma^2)$ to check whether $\mathcal{C}^{i-1}\prec\mathcal{C}^{i}$. This is since we can compare each template in $\mathcal{C}^{i-1}$ with each template in $\mathcal{C}^{i}$, there {are} $\mathcal{O}(\sigma^2)$ such pairs, and each of these comparison can be done in constant time. Furthermore, the $\min$ operator in the recursive definition \eqref{eq:pathrecursion} of $\vartheta_i[C^i]$ is taken over all the possible collections $\mathcal C^{i-1}$ such that $\mathcal C^{i-1} \prec \mathcal C^i$, so contributes time $\mathcal{O}(T^{2\sigma}n^\sigma\cdot \sigma^2)$ to the computational cost of \eqref{eq:pathrecursion}.

To conclude, computing $\vartheta_i[\mathcal C^i]$ for any fixed collection $\mathcal C^i$ at each step $i$ requires $\mathcal{O}(T^{2\sigma}n^\sigma \cdot \sigma^2 )$ time. So, doing this calculation for all feasible collections $\mathcal C^i\in \Gamma^i$ at a fixed step $i$ can be done in $\mathcal{O}(T^{4\sigma}n^{2\sigma} \cdot \sigma^2 )$ time. Furthermore, this calculation should be repeated $n$ times, one for each step of the algorithm. Thus, this dynamic programming algorithm runs in $\mathcal{O}(T^{4\sigma}n^{2\sigma +1} \cdot \sigma^2 )$ time.    \end{proof}

\section{Completion to Temporal Cliques}\label{sec:additiononly}
In this section we consider the following variant of the \EditTempClique problem where we are only allowed to add edges (i.e. the \textit{completion} of the temporal graph).

\begin{framed}
	\noindent
\CompTempClique (\CTC):\\ 
\textit{Input:} A temporal graph $\mathcal{G}=(G,\mathcal{T})$ and positive integers $k,\Delta_1,\Delta_2 \in \mathbb Z^{+}$.\\
\textit{Question:} Does there exist a set ${\Pi}$ of time-edge additions, of cardinality at most $k$, such that the modified temporal graph is a $(\Delta_1,\Delta_2)$-cluster temporal graph? 
\end{framed}

The main result of this section is to show that the above problem can be solved in time polynomial in the size of the input temporal graph.

\begin{theorem}\label{thm:polytimeadding}There is an algorithm solving \CompTempClique on any temporal graph $\mathcal{G}$ in time $\mathcal{O}\!\left(| \mathcal{E}|^3|V|\right)$. 
\end{theorem}

As observed in \cite{CE_NP}, the cluster completion problem is also solvable in polynomial time on static graphs. In this case the optimum solution is obtained by transforming each connected component of the input graph into a complete graph. The situation is not quite so simple in temporal graphs, however a similar phenomenon holds with $\Delta_2$-saturated sets taking the place of connected components.

\bigskip

\begin{proof}[Proof of Theorem \ref{thm:polytimeadding}]
Recall from Lemma \ref{sat_cliques} that a temporal graph is a $(\Delta_1,\Delta_2)$-cluster temporal graph if and only if every $\Delta_2$-saturated set of time-edges forms a $\Delta_1$-temporal clique.  It therefore suffices to identify the (unique partition into) $\Delta_2$-saturated sets, and then for each one determine the minimum number of time-edge additions required to transform it into a $\Delta_1$-temporal clique.

Following this argument, our algorithm proceeds in two stages.  The first stage involves running the algorithm of Lemma \ref{lem:find-saturated} to find the unique partition of $\mathcal{E}$ into $\Delta_2$-saturated subsets; this algorithm is easily adapted to also output the corresponding templates.  This stage takes time $\mathcal{O}(|\mathcal{E}|^3|V|)$.

In the second stage we take each template $C_i=(X_i,[a_i,b_i])$ corresponding to a $\Delta_2$-saturated set and, for each pair $\{x,y\} \subseteq X_i $, we make an ordered list $\{t_1^i,{\dots}, t_{j_{xy}} \} $ of all the time appearances of the edge $xy$ in the interval $[a_i,b_i]$, and set $t_0^{i} =a_i $ and $t_{j_{xy}+1}^{i} =b_i $ (note that we can easily do this for all templates and all pairs in time $\mathcal{O}(|\mathcal{E}|\log|\mathcal{E}|)$). 

{Observe that if, for some $i$, $\{x,y\} \subseteq X_i $, and $j\in [j_{xy}+1]$, we have $t_{j}-t_{j-1}>\Delta_{1} $ then there is a violation to the $\Delta_1$-density condition for the edge $xy$ in the clique $C_i$ between times $t_{j}$ and $t_{j-1}$. The minimum number of time-edges one must add to fix this is clearly $ \lfloor \frac{ t_{j}-t_{j-1}}{\Delta_1}\rfloor $ since adding any fewer edges leaves a gap greater than $\Delta_1$ between appearances of $xy$ in $[a_i,b_i]$. Furthermore, adding $ \lfloor \frac{ t_{j}-t_{j-1}}{\Delta_1}\rfloor $ time appearances of $xy$ as even spaced as possible between  $t_{j}$ and $t_{j-1}$ satisfies the $\Delta_1$-density condition for $xy$ in the interval $[t_{j},t_{j-1}]$. By considering all templates representing $\Delta_2$ saturated sets and all pairs of vertices in these sets, it follows that the minimum number of edges additions needed to transform $\mathcal{G}$ into a cluster temporal graph is given by} \[S= \sum_{i\in \mathcal{I}}\sum_{\{x,y\}\subseteq X_i}\sum_{j\in [j_{xy}+1]}\left\lfloor \frac{ t_{j}-t_{j-1}}{\Delta_1}\right\rfloor . \]
The algorithm then outputs $\texttt{YES}$ if $S\leq k$, and $\texttt{NO}$ otherwise. Observe that for any pair $\{x,y\}\subseteq X_i$ such that there is no time appearance of $xy$ in $[a_i,b_i]$ the contribution to the sum is just $\lfloor \frac{ b_i-a_i}{\Delta_1}\rfloor $ and thus $S$ can be computed using $\mathcal{O}\left(|\mathcal{E}|\right)$ integer operations involving integers of size at most $T$.  It is therefore clear that the running time is dominated by that of stage one, and the claimed time bound follows.  
\end{proof}

\section{A Local Characterisation of Cluster Temporal Graphs}\label{sec:characterisation} In Section \ref{sec:charproof} we give a characterisation of cluster temporal graphs. We then use this characterisation in Section \ref{sec:charalg} to give a search-tree algorithm for \ETC. 
\subsection{The Five Vertex Characterisation}\label{sec:charproof}

Recall from Remark \ref{rem:deltas} that $(\Delta_1,\Delta_2)$-cluster temporal graphs are only defined for $\Delta_2>\Delta_1\geq 1$, since otherwise the decomposition of such a graph into $\Delta_1$-temporal cliques is not unique. We give the following complete characterization for identifying $(\Delta_1,\Delta_2)$-cluster temporal graphs by sub-graphs induced by small vertex sets. We also show this is best possible in terms of the number of vertices in the induced sub-graphs.


\begin{theorem} \label{Characterization} 
Let $\mathcal G$ be any temporal graph, and $\Delta_2>\Delta_1\geq 1$. Then, the following holds.  
\begin{itemize}
    \item If $\Delta_1=1$ then $\mathcal G$ is a $(1,\Delta_2)$-cluster temporal graph if and only if $\mathcal G[S]$ is a $(1,\Delta_2)$-cluster temporal graph for every set $S\subseteq V(\mathcal G)$ of at most $3$ vertices. 
    \item If $\Delta_1 >1$ then $\mathcal G$ is a $(\Delta_1,\Delta_2)$-cluster temporal graph if and only if $\mathcal G[S]$ is a $(\Delta_1,\Delta_2)$-cluster temporal graph for every set $S\subseteq V(\mathcal G)$ of at most $5$ vertices. 
\end{itemize}

\end{theorem}

  For the special case $\Delta_1=1$ it is clear that one must consider sets of at least three vertices as this is the case for static cluster graphs. The following lemma shows through a counterexample that for $\Delta_1>1$ a characterization as above cannot be obtained by considering only subsets of at most four vertices.

\begin{lemma}\label{counter} For any $\Delta_2>\Delta_1> 1$, there exists a temporal graph on $5$ vertices which is not a $(\Delta_1,\Delta_2)$-cluster temporal graph but every induced temporal sub-graph on at most $4$ vertices is a $(\Delta_1,\Delta_2)$-cluster temporal graph.  
	\end{lemma}
\begin{proof}
Let us consider the temporal graph $\mathcal G=(G, \mathcal T)$ shown in Figure \ref{4vertexcounter} and fix any $\Delta_2>\Delta_1>1$. We claim that $\mathcal G$ is not a $(\Delta_1,\Delta_2)$-cluster temporal graph. First note that $\mathcal G$ cannot be a $\Delta_1$-temporal clique on $5$ vertices because at least the two static edges $ad$ and $be$ are missing in the underlying graph. Secondly, as we now explain, $\mathcal G$ is not a union of disjoint temporal cliques. By $\Delta_1$-density, the three time-edges $(ab,2),(bc,1)$ and $(ca,1)$ must be in the same $\Delta_1$-temporal cliques, and similarly for $(cd,\Delta_2+2),(de,\Delta_2+1)$ and $(ce,\Delta_2+2)$. However, these two cliques are not $\Delta_2$-independent because they are adjacent in the underlying graph, sharing the vertex $c$, and the time-edges $(ab,2)$ and $(de,\Delta_2+1)$ are less than $\Delta_2$ time steps apart, thus the two $\Delta_1$-temporal cliques are not $\Delta_2$-independent. Consider now any $4$-vertex subset $S$:

\textit{Case (i):} $c \in S$. In this case, for any choice of the other vertices in $S$ the underlying graph of $\mathcal G[S]$  is a triangle which includes $c$ plus an edge $cx$ where $x \in \{a,b,d,e\}$ is not in the triangle.  If $x \in \{a,b\}$ then $(cx,1)$ is $\Delta_2$-independent with respect to the triangle $cde$. On the other hand, if $x \in \{d,e\}$ then $(cx,\Delta_2+2)$ is $\Delta_2$-independent with respect to the triangle $abc$. Thus, in any case, we would obtain two $\Delta_2$-independent $\Delta_1$-temporal cliques.

\textit{Case (ii):} $c \not \in S$. In this case $\mathcal G[S]$ consists only of two disjoint time-edges; so clearly it is a $(\Delta_1,\Delta_2)$-cluster temporal graph.  
\end{proof}

\begin{figure}[h]
	\begin{center}
		\begin{tikzpicture}[xscale=.4,yscale=.4,knoten/.style={thick,circle,draw=black,fill=white},edge/.style={very thick}]
 
		\node[knoten] (a) at (-5.196,3) {a};
		\node[knoten] (b) at (-5.196,-3) {b};
		\node[knoten] (c) at (0,0) {c};
		\node[knoten] (d) at (5.196,3) {d};
		\node[knoten] (e) at (5.196,-3) {e};

		\draw[edge] (a) to (c);
		\draw[edge] (c) to (b);
		\draw[edge] (b) to (a);
		\draw[edge] (c) to (d);
		\draw[edge] (d) to (e);
		\draw[edge] (c) to (e);
		
		\node (ac) at (-2.6,2.2) { $1$};
 		\node (bc) at (-2.6,-2.3) {$1$};
		\node (dc) at (2,2.2) { $\Delta_2+2$};
		\node (dc) at (2,-2.3) {  $\Delta_2+2$};
		\node (bc) at (-5.8,0) {  $2$};
		\node (de) at (6.5,0) { $\Delta_2+1$};
		
		\end{tikzpicture}\end{center}
	\caption{The counter-example from Lemma \ref{counter}. This temporal graph shows that $4$-vertex sub-graphs are not sufficient to characterise $(\Delta_1,\Delta_2)$-cluster temporal graphs.}\label{4vertexcounter}
\end{figure}
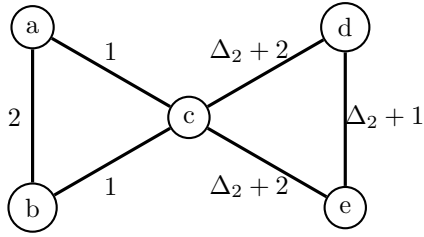

In the following we will say that a temporal graph $\mathcal{G}$ satisfies the \textit{five vertex condition} if every temporal sub-graph of $\mathcal G$ induced by (at most) five vertices is a $(\Delta_1,\Delta_2)$-cluster temporal graph.

The `only if' direction of Theorem \ref{Characterization}, which says that any $(\Delta_1,\Delta_2)$-cluster temporal graph respects the five vertex condition, is an immediate consequence of Lemma \ref{easychar}. The other `if' direction will be established in Lemmas \ref{Thm:CharhardDirection} and  \ref{lem:charDel1} and requires considerably more work. The following lemma is quite simple but illustrates a key idea in the proof of the `if' direction of Theorem \ref{Characterization}: the five vertex condition allows us to `grow' certain sets of time-edges. 

\begin{lemma} \label{lem:mallet} Let $\mathcal G$ be any temporal graph satisfying the property that $\mathcal G[S]$ is a $(\Delta_1,\Delta_2)$-cluster temporal graph for every set $S\subseteq V(\mathcal G)$ of at most five vertices. Let $\mathcal{H}'$ be a $\Delta_1$-temporal clique realising the template $(H,[c,d])$, and  $x,y\in H$. Suppose that $xy$ is $\Delta_1$-dense in the set $[a,b]\supseteq [c,d]$ and let $r_1 = \min\left(\mathcal{T}(xy)\cap [a,b]\right)$ and  $r_2 = \max\left(\mathcal{T}(xy)\cap [a,b]\right)$. Then there exists a $\Delta_1$-temporal clique $\mathcal{H}\supseteq \mathcal{H}'$ which $\Delta_1$-realises the template $(H,[r_1,r_2])$ where $[r_1,r_2] \supseteq [a+\Delta_1-1,b-\Delta_1+1]$.
\end{lemma}

\begin{proof} Fix an arbitrary pair $w,z\in H$. Because $\mathcal H'$ is a $\Delta_1$-temporal clique, the edge $wz$ is $\Delta_1$-dense in $[c,d]$. We now consider the set $A=\{x,y,w,z\}\subseteq H$. Observe that $\mathcal{H}'[A]$ is a $\Delta_1$-temporal clique and thus an indivisible set. By hypothesis $xy$ is $\Delta_1$-dense in $[a,b]$ and so if {$A_{xy}=\{(xy,t) : t\in [a,b]\}$} then this set is $\Delta_2$-indivisible. Since $x,y \in A$ and $[a,b]\supseteq [c,d] $ we have $\mathcal{H}'[A] \cap A_{xy} \neq \emptyset $, thus $\mathcal{H}'[A] \cup A_{xy}$ is $\Delta_2$-indivisible by Lemma \ref{lem:IndivisibleIntersection}. By five vertex condition we know that $\GG[A]$ is a $(\Delta_1,\Delta_2)$-cluster temporal graph, therefore by Lemma \ref{lem:indivisibility} the set $\mathcal{H}'[A] \cup A_{xy} $ must be contained within one $\Delta_1$-temporal clique $\mathcal{H}$ with lifetime containing $[r_1,r_2]$. So, since the choice of $w,z\in H $ was arbitrary, the result holds. Since $A_{xy}\subseteq \mathcal{H}$, where $A_{xy}$ is a set of edges which are $\Delta_1$-dense on $[a,b]$, it follows that $r_1\leq a+\Delta_1-1$ and $r_2 \geq b-\Delta_1+1$.
  \end{proof}

We are now ready to state and prove the `if' direction of Theorem \ref{Characterization} in the case $\Delta_1>1$. {The main idea of the proof is as follows, suppose there are at least two distinct $\Delta_1$-temporal cliques contained in a $\Delta_2$-saturated set that are not $\Delta_2$-independent. We first show that we can find a small set of time-edges that is not $\Delta_2$-independent of either clique. We then use the five vertex condition to grow this set of time-edges until it is a $\Delta_1$-temporal clique that contains both of the original $\Delta_1$-temporal cliques, a contradiction.}

\begin{lemma}\label{Thm:CharhardDirection}
Let $\mathcal{G}$ be a temporal graph and suppose that, for every set $S$ of at most five vertices in $\mathcal{G}$, the induced temporal subgraph $\mathcal{G}[S]$ is a $(\Delta_1,\Delta_2)$-cluster temporal graph.  Then $\mathcal{G}$ is a $(\Delta_1,\Delta_2)$-cluster temporal graph.
\end{lemma}

\begin{proof}
Let $\mathcal P_{\mathcal G}$ be the unique partition of $\mathcal E(\mathcal G)$ into $\Delta_2$-saturated subsets guaranteed by Lemma \ref{saturated}. Let us select any subset $S \in \mathcal P_{\mathcal G}$ and suppose that $L(S)=[s,t]$. We want to show that $S$ forms a $\Delta_1$-temporal clique in $\mathcal G$.

To prove this, we introduce a collection $\varkappa_S=\{S_1,\dots,S_m\}$ of subsets of $S$,  such that each $S_i$ is a $\Delta_1$-temporal clique for any $i \in \{1,\dots,m\}$, $S=\bigcup_{i=1}^m S_i$ and for every $S_i \in S$ there does not exist any other $\Delta_1$-temporal clique $K \subseteq  S$ such that $S_i \subsetneq K$. Note that the subsets $S_1,\dots, S_m$ in this collection are not required to be pairwise disjoint. We also observe that this collection exists: since a singleton set consisting of any time-edge is a $\Delta_1$-temporal clique, it follows that every time-edge in $S$ belongs to at least one $\Delta_1$-temporal clique. Then, we will say that each $S_i$ is a {\it maximal $\Delta_1$-temporal clique within $S$}.

We will assume for a contradiction that $m\geq 2$. Because $S$ is $\Delta_2$-saturated, it is not possible that all the $\Delta_1$-temporal cliques contained in $\varkappa_S$ are pairwise $\Delta_2$-independent, since this would imply that $S$ is not $\Delta_2$-indivisible.  Thus, as we assume $m\geq 2$, let us consider any distinct  $S_i,S_j \in \varkappa_S$ that are not $\Delta_2$-independent. We shall then show that they must both be contained within a larger $\Delta_1$-temporal clique, which itself is contained in $S$, contradicting maximality. It will then follow that $m=1$ and thus $S$ is itself a $\Delta_1$-temporal clique, which establishes the lemma. 

The next claim shows that if two maximal $\Delta_1$-temporal cliques in $S$ are not $\Delta_2$-independent, then there is a small sub-graph witnessing this non-independence.

\begin{clm} \label{W}
Let $S_i,S_j \in \varkappa_S$ be any pair of $\Delta_1$-temporal cliques which are not $\Delta_2$-independent. Then, there exists a set $W \subseteq V(S_i) \cup V(S_j)$ with $|W| \leq 5$ such that $W$ contains a vertex $c\in V(S_i)\cap V(S_j)$.{ Additionally,} $(S_i \cup S_j)[W]$ is $\Delta_2$-indivisible and contains the time-edges $(xc,r_i)\in S_i$ and $(zc,r_j) \in S_j$, for some $x,z\in V$. 
\end{clm}

\begin{poc}[\ref{W}] As $ S_i$ and $S_j$ are $\Delta_1$-temporal cliques, it follows that $ S_i[W]$ and $S_j[W]$ are $\Delta_1$-temporal cliques for any $W\subseteq V$ by Lemma \ref{easychar}. Since  $S_i$ and $S_j$ are chosen so that they are not $\Delta_2$-independent, they must share at least one vertex $c$ and the following cases can occur:

{\it Case 1:} there exist respectively two time-edges $(xc,r_i) \in S_i$ and $(zc,r_j) \in S_j$ such that $|r_i-r_j| < \Delta_2$. For this case set $W=\{c,x,z \}$. The set $I=\{(cx,r_i),(cz,r_j)\}$ is $\Delta_2$-indivisible and has non-empty intersection with both $S_i[W]$ and $S_j[W]$. Thus by Lemma \ref{lem:IndivisibleIntersection}, $I\cup  S_i[W]$ and $I\cup  S_j[W]$ are both $\Delta_2$-indivisible. Applying Lemma \ref{lem:IndivisibleIntersection} again gives that $S_i[W]\cup S_j[W] = (S_i \cup S_j)[W]$ is $\Delta_2$-indivisible.

{\it Case 2:} there exist respectively two time-edges $(xy,t) \in S_i$ and $(zw,t') \in S_j$ such that $|t-t'| < \Delta_2$, (incidentally, Figure \ref{4vertexcounter} is an example of this kind of structure). Let us note that, because $c \in V(S_i) \cap V(S_j)$, then $x,y,z$ and $w$ must all be adjacent to the vertex $c$ in the underlying graph. Set $W=\{c,x,y,z,w\} \subset V(S_i) \cup V(S_j)$. Let $A_c$ be the set of appearances of the edges $cx,cy,cz$ and $cw$ in $L(S_i) \cup L(S_j)$, thus $A_c \cup \{(xy,t),(zw,t')\} \subseteq (S_i \cup S_j)[W]$ is $\Delta_2$-indivisible and intersects $S_i[W]$ and $S_j[W]$. So, as before, applying Lemma \ref{lem:IndivisibleIntersection} gives that $S_i[W]\cup S_j[W] = (S_i \cup S_j)[W]$ is $\Delta_2$-indivisible. Furthermore, since $x,y,c \in V(S_i)$ and $S_i$ is a $\Delta_1$-temporal clique, there must exist at least one time-edge $(xc,r_i)\in S_i$ that is at distance at most $\Delta_1<\Delta_2$ from $(xy,t)$, thus $(xc,r_i)\in (S_i \cup S_j)[W] $ by Lemma \ref{lem:IndivisibleIntersection}. The same applies for some $(cz,r_j)\in S_j$.\end{poc}

Recall that $S_i$ and $S_j$ are both $\Delta_2$-indivisible as they are $\Delta_1$-temporal cliques, and that by Claim \ref{W} there exists a set $W\subseteq V$ such that $(S_i \cup S_j)[W]$ is $\Delta_2$-indivisible and has non-empty intersection with $S_i$ and $S_j$. It follows from Lemma \ref{lem:IndivisibleIntersection} that both $S_i[W]\cup S_j$ and $S_j[W]\cup S_i$ are $\Delta_2$-indivisible. As their intersection is $(S_i \cup S_j)[W]\neq \emptyset$, invoking Lemma \ref{lem:IndivisibleIntersection} once again gives that $S_i\cup S_j$ is $\Delta_2$-indivisible. Recall that $L(S)=[s,t]$.  
 
\begin{clm} \label{lem:weakened}
Let $S_i$ and $S_j$ {be as in Claim \ref{W}}, where $L(S_i)=[s_i,t_i]$, and $L(S_j)=[s_j,t_j]$. Then, there exists some $K\subseteq V$ and $s',t'\in \mathbb{Z}^+$ such that $\mathcal G$ contains a $\Delta_1$-temporal clique $\mathcal K $ which generates the template $(K,[s',t'])$ where $s' \in [s, \min \{s_i,s_j\}+\Delta_1-1]$, $t' \in [\max \{t_i,t_j\}-\Delta_1+1,t]$. Furthermore, $(xc,r_i)  \in S_i\cap \mathcal{K}$, and $(cz,r_j)  \in S_j\cap \mathcal{K}$, where the time-edges $(xc,r_i)$ and $(cz,r_j)$ are specified by Claim \ref{W}.
\end{clm}

 \begin{poc}[\ref{lem:weakened}]By Claim \ref{W}, there exists a set $W\subseteq V$ of at most five vertices and two time-edges $(xc,r_i) \in S_i$ and $(cz,r_j) \in S_j$ such that $(S_i \cup S_j)[W] $ is $\Delta_2$-indivisible and $(xc,r_i),(cz,r_j)\in (S_i \cup S_j)[W]$. By the five vertex condition $\mathcal{G}[W]$ is $(\Delta_1,\Delta_2)$-cluster temporal graph. Thus, we can apply Lemma \ref{lem:indivisibility} to see that  $(S_i \cup S_j)[W] $ is contained within a single $\Delta_1$-temporal clique $\mathcal{K}'$. Since $(xc,r_i), (cz,r_j)\in (S_i \cup S_j)[W] $, we see that the lifetime $L(\mathcal{K}')=[k_1,k_2]$ of this $\Delta_1$-temporal clique intersects both $L(S_i)$ and $L(S_j)$. 
 
 The edge $xc$ is $\Delta_1$-dense in both $S_i$ and $\mathcal{K}'$. Since $(xc,r_i) \in S_i\cap \mathcal{K}'$, we can further deduce that the edge $xc$ is $\Delta_1$-dense in $L(S_i) \cup L(\mathcal{K}') = \left[\min\{s_i,k_1 \}, \max\{t_i,k_2\} \right] $. Applying  Lemma \ref{lem:mallet} to $\mathcal{K}'$ and $xc$, where we take $[a,b] = L(S_i) \cup L(\mathcal{K}') $ and $[c,d]= L(\mathcal{K}') $, we can extend $L(\mathcal{K}')$ to an interval $[r_1^{xc},r_2^{xc}] $ satisfying
\begin{equation}\label{eq:xyextend}[r_1^{xc},r_2^{xc}] \supseteq \left[\min\{s_i,k_1 \} + \Delta_1 -1 , \max\{t_i,k_2\} -\Delta_1 +1 \right].\end{equation}

 Similarly, $cz$ is $\Delta_1$-dense in $L(S_j) \cup L(\mathcal{K}') = \left[\min\{s_j,k_1 \}, \max\{t_j,k_2\} \right] $. Applying Lemma \ref{lem:mallet} to $\mathcal{K}'$ and $cz$ with $[a,b] = L(S_j) \cup L(\mathcal{K}') $ and $[c,d]= L(\mathcal{K}') $ extends the lifetime of $\mathcal{K}'$ to an interval $[r_1^{cz},r_2^{cz}] $ satisfying
\begin{equation}\label{eq:wzextend}[r_1^{cz},r_2^{cz}] \supseteq \left[\min\{s_j,k_1 \} + \Delta_1 -1 , \max\{t_j,k_2\} -\Delta_1 +1 \right].\end{equation}

It follows that there exists some $\Delta_1$-temporal clique $\mathcal{K}$ in $\mathcal{G}$ with vertex set $K$ and lifetime $L(\mathcal{K})$ satisfying $L(\mathcal{K})\supseteq [\min\{r_1^{xc},r_1^{cz}\}, \max\{r_2^{xc},r_2^{cz}\}  ] $. Observe that from \eqref{eq:xyextend} and \eqref{eq:wzextend} we have  
\begin{align*} \min\{r_1^{xc},r_1^{cz}\} & = \min\left\{ \min\{s_i,k_1 \} + \Delta_1 -1 , \min\{s_j,k_1 \} + \Delta_1 -1 \right\}\\
& =\min\left\{ \min\{s_i,k_1 \}   , \min\{s_j,k_1 \}  \right\}  + \Delta_1 -1 \\
&\leq \min\left\{ s_i    , s_j  \right\}  + \Delta_1 -1 ,\end{align*} and similarly $\max\{r_2^{xc},r_2^{cz}\} \geq \max\left\{ t_i    , t_j  \right\}  - \Delta_1 +1 $.  Finally, since $S$ is $\Delta_2$-saturated, it follows by Lemma \ref{sat_set} that $L(\mathcal K) \subseteq L(S) = [s,t]$. \end{poc}

Let us now consider {$\mathcal K=(K,[s',t'])$}, the $\Delta_1$-temporal clique of Claim \ref{lem:weakened}. From this we want to extend $S_i$ and $S_j$ to a $\Delta_1$-temporal clique with vertex set $V(S_i) \cup V(S_j)$ and lifetime at least $L(\mathcal K) \cup L(S_i) \cup L(S_j)=[\min \{s_i,s_j,s'\},\max \{t_i,t_j,t'\}]$. We achieve this over the following two claims.

\begin{clm} \label{cl:hclique}Let $\mathcal K=(K,[s',t'])$ be given by Claim \ref{lem:weakened}. Then there exist $h_1 \leq h_2$ satisfying \[[h_1,h_2] \supseteq [\min\{s_i,s_j,s'\} + \Delta_1-1,\max\{t_i,t_j,t'\} - \Delta_1+1  ],\] and a $\Delta_1$-temporal clique in $\mathcal{G}$ generating the template $(V(S_i) \cup V(S_j),[h_1,h_2])$. Furthermore, this  $\Delta_1$-temporal clique contains the two edges  $(xc,r_i)\in S_i\cap \mathcal{K} $ and $(cz,r_j)\in S_j\cap \mathcal{K} $, specified by Claim \ref{W}. 
\end{clm}

\begin{poc}[\ref{cl:hclique}] By Claim \ref{lem:weakened}  the $\Delta_1$-temporal clique $\mathcal{K}$ contains the time-edges $(xc,r_i) \in S_i$ and  $(cz, r_j) \in S_j$. Let $A_{xc}$ and  $A_{cz}$ be the set of appearances of $xc$ in $L(S_i) \cup L(\mathcal{K})$ and $cz$ in $L(S_j) \cup L(\mathcal{K})$ respectively. The edge $xc$ is $\Delta_1$-dense in both $L(S_i)$ and $L(\mathcal{K})$, thus since $(xc,r_i)  \in S_i\cap \mathcal{K}$, it follows that $xc$ is $\Delta_1$-dense in $L(S_i) \cup L(\mathcal K)$. Similarly, $cz$ is $\Delta_1$-dense in $L(S_j) \cup L(\mathcal K)$. Since both $A_{xc}$ and $ A_{cz}$ contain a set of edges that is dense in $L(\mathcal{K})$, we see that $A_{xc} \cup A_{cz}$ must be $\Delta_2$-indivisible. Let $r_1$ and $r_2$ respectively be the earliest and latest times in $L(S_i) \cup L(S_j) \cup L(\mathcal{K})$ at which either $xc$ or $cz$ appear. Note that as $xc$ is $\Delta_1$-dense in $L(S_i) \cup L(\mathcal{K})$ we must have $r_1\leq \min\{s_i,s'\}+\Delta_1-1$. Applying the same reasoning to $cz$ we have $r_1\leq \min\{s_j,s'\}+\Delta_1-1$, thus $r_1\leq \min\{s_i,s_j,s'\}+\Delta_1-1$. By a symmetric argument we also see that $r_2\geq \max\{t_i,t_j,t'\}-\Delta_1+1$.

Let us fix an arbitrary pair of vertices $u,v \in V(S_i) \cup V(S_j)$.  We want to show that $uv$ is $\Delta_1$-dense in an interval containing $[r_1,r_2]$.  Suppose first that $u \in V(S_i)$ and $v \in V(S_j)$.  Since $S_i$ and $S_j$ are $\Delta_1$-temporal cliques, we know that the edge $ux$ exists in the underlying graph of $S_i$ and is $\Delta_1$-dense in $L(S_i)$; similarly $zv$ is present in the underlying graph of $S_j$ and is $\Delta_1$-dense in $L(S_j)$.  Let $A_{ux}$ be the set of appearances of $ux$ in $L(S_i)$ and $A_{zv}$ the set of appearances of $zv$ in $L(S_j)$.  Note that, because both the $\Delta_2$-indivisible sets $A_{ux}$ and $A_{xc}$ are $\Delta_1$-dense in $L(S_i)$, then $A_{ux} \cup A_{xc}$ is $\Delta_2$-indivisible; similarly $A_{zv} \cup A_{cz}$ is $\Delta_2$-indivisible too. Recall that $ A_{xc} \cup A_{cz}$ is also $\Delta_2$-indivisible.  It follows that $A_{ux} \cup A_{xc} \cup A_{cz} \cup A_{zv}$ is an $\Delta_2$-indivisible set by two applications of Lemma \ref{lem:IndivisibleIntersection}. Let us consider the temporal graph $\mathcal G[N]$ induced by the set $N=\{u,x,c,z,v\}$.  By the five vertex condition, we know that this must form a $(\Delta_1,\Delta_2)$-cluster temporal graph and hence, since the set $A_{ux} \cup A_{xc} \cup A_{cz} \cup A_{zv}$ is $\Delta_2$-indivisible, the restriction of $\mathcal{G}[N]$ to the interval $L(S_i) \cup L(S_j) \cup L(\mathcal{K})$ must form a single $\Delta_1$-temporal clique by Lemma \ref{lem:indivisibility}.  The lifetime $L_{u,v}$ of this $\Delta_1$-temporal clique contains the interval $[r_1,r_2]$, and every edge appears at least once within $\Delta_1$ timesteps of $r_1$ and once within $\Delta_1$ timesteps of $r_2$. 
For the remaining cases where $u,v \in V(S_i)$ or $u,v \in V(S_j)$ we can use a  similar argument. However if $u,v \in V(S_i)$ it is enough to consider indivisibility of the set $A_{xc} \cup A_{ux}\cup A_{xv}$, and if $u,v \in V(S_i)$ one considers $A_{cz} \cup A_{uz}\cup A_{zv}$, so the arguments are simpler. 

Up to this point we have shown that for any $u,v\in V(S_i)\cup V(S_j)$ the edge $uv$ is $\Delta_1$-dense in an interval $L_{u,v}$ containing $[r_1,r_2]$. We now argue that there is an interval $[h_1,h_2]$ such that, for every pair $u,v \in V(S_i) \cup V(S_j)$, the edge $uv$ is $\Delta_1$-dense in $[h_1,h_2]$.  If every such edge has at least one appearance in the interval $[r_1,r_2]$, then we can set $[h_1,h_2] = [r_1,r_2]$ and we are done. In particular, we note that this condition is certainly respected if $|r_2-r_1| \geq \Delta_1-1$. Otherwise, it could happen that there exists another pair of vertices $u',v' \in V(S_i) \cup V(S_j)$ such that $|L_{u,v} \cap L_{u',v'}|< \Delta_1-1$, in which case there might be no interval in which both the edges $uv$ and $u'v'$ are $\Delta_1$-dense.

Let us assume for a contradiction that there is no interval containing $[r_1,r_2]$ in which every edge with both endpoints in $V(S_i) \cup V(S_j)$ is $\Delta_1$-dense. Recall that this can only happen if $r_2 < r_1 + \Delta_1 - 1$. Moreover we can assume there does not exist an interval $[r_1',r_2'] \supseteq [r_1,r_2]$ with $r_2' \le r_1' + \Delta_1 - 1$ such that every edge $u'v'$, where $u',v' \in V(S_i)\cup V(S_j)$, has at least one appearance in $[r_1',r_2']$.  We now fix two edges $uv$ and $u'v'$, with $u,v,u',v' \in V(S_i) \cup V(S_j)$, such that there is no interval containing $[r_1,r_2]$ in which both $uv$ and $u'v'$ are $\Delta_1$-dense. Recall from earlier reasoning that any edge $e$ with both endpoints in $ V(S_i)\cup V(S_j)$ is $\Delta_1$-dense in the interval $L_{e}$ which contains $[r_1,r_2]$ and moreover, that  there are appearances of $e$ within $L_{e}$ which are also within $\Delta_1$ time steps of $r_1$ and $r_2$ respectively. We claim that each such $e$ appears at least once in the interval $I=[r_2 - \Delta_1+1, r_1 + \Delta_1-1]$. Observe that this interval contains at least $\Delta_1 +1$ and at most $2\Delta_1-1$ time steps by the condition $r_2<r_1 +\Delta_1 -1 $. Thus the interval must contain at least one time appearance of $e$ since there are appearances of $e$ in $L_{e}$ at or before time $r_1 + \Delta_1 - 1$ and at or after time $r_2 - \Delta_1 + 1$. For some edge $e$ let $I_e$ be the set of appearances of $e$ in $I =[r_2 - \Delta_1+1, r_1 + \Delta_1-1]$. To find a contradiction we prove that $uv$ and $u'v'$ belong to a $\Delta_1$-temporal clique whose lifetime contains $[r_1,r_2]$.

  We will begin by showing that $I_{uv}\cup I_{uc}\cup I_{vc} $ and  $I_{u'v'}\cup I_{u'c}\cup I_{v'c}$ are both $\Delta_2$-indivisible. Recall that  $u,v,c$ belong to a $\Delta_1$-temporal clique whose lifetime contains the interval $[r_1,r_2]$, and that each edge in this clique must appear at least once in the time interval $I$. Recall also that $cx$ appears at some time  $r_i\in [r_1,r_2]$ and thus $cx$ is $\Delta_1$-dense in $I$. Thus, no appearance of $uc$ and $vc$ in $I$ can be $\Delta_2$-independent of $I_{cx}$ and so $I_{uc}\cup I_{vc} \cup I_{cx} $ is $\Delta_2$-indivisible. Hence, the set $I_{uc}\cup I_{vc} \cup I_{cx}$ is contained in a single $\Delta_1$-temporal clique $\mathcal{C}$ of the $(\Delta_1,\Delta_2)$-cluster temporal graph $\mathcal{G}[\{u,v,c,x\}]$  by Lemma \ref{lem:indivisibility} together with the five vertex condition. Any appearance of $uv$ in $I$ must also be contained in $\mathcal{C}$, since it cannot be $\Delta_2$-independent of $\mathcal{C}$, and $\mathcal{G}[\{u,v,c,x\}]$ is a $(\Delta_1,\Delta_2)$-cluster temporal graph. It follows that $I_{uv}\cup I_{uc}\cup I_{vc} $ is $\Delta_2$-indivisible since all these edges are contained within the same clique $\mathcal{C}$ in $\mathcal{G}$. A symmetric argument shows that $I_{u'v'}\cup I_{u'c}\cup I_{v'c}$ is $\Delta_2$-indivisible.

 We now show  and $I_{cu}\cup I_{cu'}$ is $\Delta_2$-indivisible.  As before, $u,u',c$ belong to a $\Delta_1$-temporal clique whose lifetime contains the interval $[r_1,r_2]$, and each edge in this clique appears at least once in the time interval $I$. As $cx$ appears at some time  $r_i\in [r_1,r_2]$, no appearance of $uc$ and $u'c$ in $I$ can be $\Delta_2$-independent of $I_{cx}$, thus $I_{uc}\cup I_{u'c} \cup I_{cx} $ is $\Delta_2$-indivisible. Hence $I_{uc}\cup I_{u'c} \cup I_{cx} $ is contained within a single  $\Delta_1$-temporal clique in $\mathcal{G}[\{c,u,u'\}]$ by Lemma \ref{lem:indivisibility} and the five vertex condition. It follows that $I_{cu}\cup I_{cu'}$ is $\Delta_2$-indivisible.

We have shown that $I_{uv}\cup I_{uc}\cup I_{vc} $, $I_{u'v'}\cup I_{u'c}\cup I_{v'c}$ and $I_{cu}\cup I_{cu'}$  are $\Delta_2$-indivisible, so two applications of Lemma \ref{lem:IndivisibleIntersection} give us that $M= I_{uv}\cup I_{uc}\cup I_{vc}\cup I_{u'v'}\cup I_{u'c}\cup I_{v'c}$ is $\Delta_2$-indivisible. It then follows from Lemma \ref{lem:indivisibility} that all edges of $M$ are contained in the same $\Delta_1$-temporal clique in $\mathcal{G}[\{u,v,u',v',c\}]$. Since this clique contains the non-empty sets $I_{uv}$ and $I_{u'v'}$ it follows that there exists an interval contained in $I$ in which both $uv$ and $u'v'$ are $\Delta_1$-dense. Since both endpoints of $I$ are at distance at most $\Delta_1$ from $r_1$ and $r_2$ there must be an interval containing $[r_1,r_2]$ in which both $uv$ and $u'v'$ are $\Delta_1$-dense, a contradiction. 

We can therefore conclude that there is some interval \[[h_1,h_2] \supseteq [r_1,r_2] \supseteq [\min\{s_i,s_j,s'\} + \Delta_1-1,\max\{t_i,t_j,t'\} - \Delta_1+1  ],\] such that every edge with both endpoints in $V(S_i) \cup V(S_j)$ is $\Delta_1$-dense in $[h_1,h_2]$, as required.\end{poc}

\begin{clm}\label{clm:cliqueextend}There is a $\Delta_1$-temporal clique in $\mathcal{G}$ generating 
$(V(S_i) \cup V(S_j), [\min\{s_i,s_j,h_1\},\max\{t_i,t_j,h_2\}])$.
\end{clm}

\begin{poc}[\ref{clm:cliqueextend}] Let $u,v\in V(S_i)\cup V(S_j)$; it suffices to demonstrate that $uv$ is $\Delta_1$-dense in the interval $[\min\{s_i,s_j\},\max\{t_i,t_j\}]$, as each edge $uv$ is $\Delta_1$-dense in $[h_1,h_2]$ by Claim \ref{cl:hclique}. Indeed,  by Claim \ref{cl:hclique}, if $h_1\leq \min\{s_i,s_j\} $ and $h_2\geq \max\{t_i,t_j\} $ then the result holds. We thus assume that at least one of these conditions does not hold and begin by considering the lower end of the interval $[\min\{s_i,s_j,h_1\},\max\{t_i,t_j,h_2\}]$.

	Assume that $h_1> \min\{s_i,s_j\} $, then by the definition of $s_i,s_j$ there exists at least one edge $ab$ with $a,b \in V(S_i) \cup V(S_j)$ that appears at time $\min\{s_i,s_j\}$. We can assume that without loss of generality that $s_i\leq s_j$, thus $s_i= \min\{s_i,s_j\}$, and also that $(ab, s_i) \in L(S_i) $. 	
	
 For the first step we recall that the edge $cx$ is $\Delta_1$-dense in both $L(S_i)$ and $[h_1,h_2]$ and, by Claim \ref{cl:hclique},  $(cx,r_i)  \in S_i$ and $r_i\in [h_1,h_2]$. Thus $cx$ is $\Delta_1$-dense in $L(S_i) \cup [h_1,h_2]$.

For the second step, recall that $ab$ is $\Delta_1$-dense in $L(S_i)$ and $[h_1,h_2]$, and appears at time $s_i$. Observe that $ax,bx,ac$ and $bc$ must all have time appearances in $[h_1,h_2]$ by Claim \ref{cl:hclique}. Thus the set of time-edges in $L(S_i)\cup [h_1,h_2]$ with endpoints in $\{a,b,c,x\}$ forms a $\Delta_2$-indivisible set and hence, by Lemma \ref{lem:indivisibility} and the five vertex condition, all these edges must be part of the same $\Delta_1$-temporal clique in $\mathcal{G}[\{a,b,c,x\}]$. It follows that $ab$ is $\Delta_1$-dense in $L(S_i) \cup [h_1,h_2]=[s_i, h_2]$. 

For the third step we consider the set $A=\{a,b,u,v\}$. By Claim \ref{cl:hclique} the set of time-edges with both endpoints in $A$ which occur in $[h_1,h_2]$ forms a $\Delta_1$-temporal clique, and thus is  $\Delta_2$-indivisible. Since $ab$ is $\Delta_1$-dense in $[s_i, h_2]$, it follows that the set of time-edges appearing in $[s_i, h_2] $ with both endpoints in $A$ is also $\Delta_2$-indivisible. By the five vertex condition and Lemma \ref{lem:indivisibility}, all these edges must belong to the same clique in $\mathcal{G}[A]$. Thus, since $(ab,s_i)\in S_i $ and $uw$ appears in $[h_1,h_2]$, it follows that $uw$ is $\Delta_1$-dense in $[s_i, h_2]$.

	Similarly, assuming $h_2<\max\{t_i,t_j\} $, there exists at least one edge $a'b'$ with $a',b' \in V(S_i) \cup V(S_j)$ that appears at time $\max\{t_i,t_j\}$. We now follow the same steps as before. For the first step, if $t_i\geq t_j$, then we take $cx$ as this is $\Delta_1$-dense in $L(S_i) \cup [h_1,h_2]$ and appears in $L(S_i)$. Otherwise, if $t_i<t_j$, we take $cz$ which is $\Delta_1$-dense in $L(S_j) \cup [h_1,h_2]$ and appears in $L(S_j)$. Once we have completed the first step to find some edge $fg$ (which is either $cx$ or $cz$ depending on whether $s_i$ or $s_j$ is larger) that is dense in a suitable interval, we complete the second step by using  $\mathcal{G}[\{a',b',f,g\}]$ to show that $a'b'$ is $\Delta_1$-dense in $[ h_1, \max\{t_i,t_j \}]$. We can then apply the third step to the set $A'=\{a',b',u,v\}$ to show that $uv$ is $\Delta_1$-dense in $[h_1, \max\{t_i,t_j\}]$. 
	
Finally, since $uv$ appears in $[h_1,h_2]$ it follows that $uv$ is $\Delta_1$-dense in  $ [\min\{s_i,s_j,h_1\},\max\{t_i,t_j,h_2\}]$.  \end{poc}

Observe that Claim \ref{clm:cliqueextend} contradicts the initial assumption that $S_i$ and $S_j$ were maximal in $S$. Hence the assumption that $m\geq 2$ must be incorrect and thus $S$ consists of a single $\Delta_1$-temporal clique. Because $S$ was a generic set of the partition $\mathcal P_{\mathcal G}$ of the given temporal graph $\mathcal G$ into $\Delta_2$-saturated subsets, then $\mathcal{G}$ must be a $(\Delta_1,\Delta_2)$-cluster temporal graph by Lemma \ref{sat_cliques}, so the statement of the lemma holds.   \end{proof}

It remains to prove that, in the special case where $\Delta_1=1$, three vertices are enough to detect whether $\mathcal{G}$ is a $(1,\Delta_2)$-cluster temporal graph. The proof will follow the same high level idea as the proof of Lemma \ref{Thm:CharhardDirection}, however it will be significantly simpler.  

\begin{lemma}\label{lem:charDel1}
Let $\mathcal{G}$ be a temporal graph and suppose that, for every set $S$ of at most three vertices in $\mathcal{G}$, the induced temporal subgraph $\mathcal{G}[S]$ is a $(1,\Delta_2)$-cluster temporal graph.  Then $\mathcal{G}$ is a $(1,\Delta_2)$-cluster temporal graph.
\end{lemma}
\begin{proof}
We shall follow the same structure as the proof of Lemma \ref{Thm:CharhardDirection}. In particular, we shall assume that within the unique decomposition of $\mathcal{E}(\mathcal{G})$ into $\Delta_2$-saturated subsets, there exists a saturated set $S$ that is not a $\Delta_1$-temporal clique. We then have a decomposition of $S$ into maximal $\Delta_1$-temporal cliques $\{S_1, \dots, S_m\}$. It suffices to demonstrate that, if there exist two $1$-temporal cliques $S_i$ and $S_j$ in $\mathcal{G}$ that are not $\Delta_2$-independent, then there must be some clique $S'$ with $S_i, S_j \subseteq S'$, since this contradicts the maximality of $S_i$ and $S_j$. As in the proof of Lemma \ref{Thm:CharhardDirection}, we suppose that $L(S_i) = [s_i,t_i]$ and $L(S_j) = [s_j,t_j]$, and note that, by definition of $1$-density, every edge in $S_i$ (respectively $S_j$) appears at every time in $[s_i,t_i]$ (respectively $[s_j,t_j]$).

Since $S_i$ and $S_j$ are not $\Delta_2$-independent, they must share a common vertex $c$.  We first argue that every edge of the form $cz$, where  $z\in (V(S_i) \cup V(S_j))\backslash \{c\}$, must be $1$-dense in $L := [\min\{s_i,s_j\},\max\{s_j,t_j\}]$.  For any edge $e$ let $A_{e}$ be the sets of appearances of $e$ in $L$. We now claim that for any vertices $x \in V(S_i)$ and $y \in V(S_j)$, the set $A_{cx} \cup A_{cy}$ is $\Delta_2$-indivisible. This follows since $cx$ and $cy$ appear at all times in $[s_i,t_i]$ and $[s_j,t_j]$ respectively, and there exists $t\in [s_i,t_i]$ and $t'\in [s_j,t_j]$ such that $|t-t'|<\Delta_2$. By hypothesis $\mathcal{G}[\{c,x,y\}]$ is a $(1,\Delta_2)$-cluster temporal graph and thus, by Lemma \ref{lem:indivisibility}, the set $A_{cx} \cup A_{cy}$ of time-edges must contained within a single $1$-temporal clique, with lifetime which contains the interval
\begin{align*}\left[\min\{t: (cx,t) \in A_{cx} \text{ or } (cy,t) \in A_{cy}\},\, \max\{t: (cx,t) \in A_{cx} \text{ or } (cy,t) \in A_{cy}\}\right] &= [\min\{s_i,s_j\},\,\max\{s_j,t_j\}] \\&= L. \end{align*}
Since $x\in V(S_i)$ and $y\in V(S_j)$ were arbitrary, the edge $cz$ is $1$-dense in $L$ for any $z\in V(S_i)\cup V(S_j)$.

Now fix two arbitrary vertices $u,v \in V(S_i) \cup V(S_j) \setminus \{c\}$; we wish to demonstrate that $uv$ is also $1$-dense in $L$, which will complete the proof.  Recall that we have just shown that $uc$ and $vc$ are $1$-dense in $L$. It follows immediately that $A_{uv} \cup A_{vc}$ is $\Delta_2$-indivisible. Thus, by the three-vertex condition, there must be a single $1$-temporal clique containing $u$, $v$ and $c$, whose lifetime contains the non-empty interval $L$.  

It follows that every edge $uv$, where $u,v\in V(S_i)\cup V(S_j)$ is $1$-dense in an interval containing $L$ and, since $uv$ therefore appears at every time during this interval. Thus there exists a $\Delta_1$-temporal clique $S'$ in $\mathcal{G}$ generating the template $(V(S_i)\cup V(S_j), L) $. This contradicts the assumption that $S_i$ and $S_j$ were maximal. 
\end{proof}

\subsection{A Search-Tree Algorithm}\label{sec:charalg}

Using the characterisation from the previous section, we are now able to prove the following result.

\begin{theorem} \ETC can be solved in time $(10 T)^k \cdot T^3 |V|^5$.
\end{theorem}
\begin{proof}
We describe a simple search tree strategy; the depth of the search tree is bounded by $k$, and the degree is bounded by $10 T$.  Since the approach is very standard, we only give an outline of the proof here.

By Theorem \ref{Characterization}, we know that a temporal graph is a $(\Delta_1,\Delta_2)$-cluster temporal graph if and only if every subset of five vertices induces a $(\Delta_1,\Delta_2)$-cluster temporal graph.  By Lemma \ref{prop:verify-poly}, given a single set of five vertices, we can verify in time $\mathcal{O}(T^3)$ whether it induces a cluster temporal graph, and the number of possible five vertex subsets is $\mathcal{O}(|V|^5)$.  Thus, in time $\mathcal{O}(T^3|V|^5)$, we can either determine that $\mathcal{G}$ is already a $(\Delta_1,\Delta_2)$-cluster temporal graph, or else identify a set $W$ of five vertices such that $\mathcal{G}[W]$ is not a $(\Delta_1,\Delta_2)$-cluster temporal graph.  In order to modify $\mathcal{G}$ to be a $(\Delta_1,\Delta_2)$-cluster temporal graph, at least one time-edge with both endpoints in $W$ must be added or deleted; as there are at most {$\binom{5}{2}=10$} possibilities for the choice of static edge and at most $T$ possibilities for the time at which the modified edge occurs (by Lemma \ref{prop:within-lifetime}), the number of possible modifications to consider is at most $10 T$.  This bounds the degree of the search tree; since the total number of modifications cannot be more than $k$, the depth is bounded by $k$.  It follows that the total number of nodes in the search tree is at most $(10 T)^k$.
  \end{proof}

\section{Conclusions and Open Problems}
In this paper we introduced a new temporal variant of the cluster editing problem, \ETC, based on a natural interpretation of what it means for a temporal graph to be divisible into `clusters'. We showed hardness of this problem even in the presence of strong restrictions on the input, but identified two special cases in which polynomial-time algorithms exist: firstly, if {the} underlying graph is a path and the number of appearances of each edge is bounded by a constant, and secondly if we are only allowed to add (but not delete) time-edges.  One natural open question arising from the first of these positive results is whether bounding the number of appearances of each edge can lead to tractability in a wider range of settings: we conjecture that the techniques used here can be generalised to obtain a polynomial-time algorithm when the underlying graph has bounded pathwidth, and it may be that they can be extended even further.   

Our main technical contribution was Theorem \ref{Characterization}, which gives a characterisation of $(\Delta_1,\Delta_2)$-cluster temporal graphs in terms of five vertex subsets.  In addition to providing substantial insight into the structure of $(\Delta_1,\Delta_2)$-cluster temporal graphs, Theorem~\ref{Characterization} also gives rise to a simple search tree algorithm, which is an \FPT algorithm parameterised simultaneously by the number $k$ of permitted modifications and the lifetime of the input temporal graph.  An interesting direction for further research would be to investigate whether this result can be strengthened: does there exist a polynomial {problem} kernel with respect to this {combined} parameterisation, and is \ETC in \FPT parameterised by the number of permitted modifications alone?

\section*{Acknowledgements}
Meeks and Sylvester gratefully acknowledge funding from the Engineering and Physical Sciences Research Council (ESPRC) grant number EP/T004878/1 for this work, while Meeks was also supported by a Royal Society of Edinburgh Personal Research Fellowship (funded by the Scottish Government).

\bibliography{ref}
\bibliographystyle{plain}

\end{document}